\documentclass{article}

\usepackage[final, nonatbib]{neurips_2020}

\usepackage[utf8]{inputenc} % allow utf-8 input
\usepackage[T1]{fontenc}    % use 8-bit T1 fonts
\usepackage{hyperref}       % hyperlinks
\usepackage{url}            % simple URL typesetting
\usepackage{booktabs}       % professional-quality tables
\usepackage{amsfonts}       % blackboard math symbols
\usepackage{nicefrac}       % compact symbols for 1/2, etc.
\usepackage{microtype}      % microtypography

\usepackage[]{algorithm}
\usepackage[]{algorithmic}
\usepackage{mathtools}

\usepackage{bm}
\usepackage{bbm}
\usepackage{amsthm}
\usepackage{amsfonts}
\usepackage{amssymb}
\usepackage{MnSymbol}
\usepackage{tabularx}
\usepackage[square,numbers]{natbib}
\usepackage{tikz}
\usetikzlibrary{matrix}
\newtheorem{theorem}{Theorem} 
\newtheorem{postulate}{Postulate}
\newtheorem{corollary}{Corollary} 
\newtheorem{lemma}{Lemma} 
\theoremstyle{definition}
\newtheorem{remark}{Remark}
\newtheorem{definition}{Definition}

\title{Universal Function Approximation on Graphs} % using Multivalued Functions

\author{%
  Rickard Br\"uel-Gabrielsson \\
  \texttt{rbg@cs.stanford.edu} \\
}

\begin{document}

\maketitle

\begin{abstract}
  In this work we produce a framework for constructing universal function approximators on graph isomorphism classes. We prove how this framework comes with a collection of theoretically desirable properties and enables novel analysis. We show how this allows us to achieve state-of-the-art performance on four different well-known datasets in graph classification and separate classes of graphs that other graph-learning methods cannot. Our approach is inspired by persistent homology, dependency parsing for NLP, and multivalued functions. The complexity of the underlying algorithm is $O(\text{\#edges} \times \text{\#nodes} )$ and code is publicly available\footnote{\url{https://github.com/bruel-gabrielsson/universal-function-approximation-on-graphs}}.
\end{abstract}

%%%%%%%%% BODY TEXT
\section{Introduction}

Graphs are natural structures for many sources of data, including molecular, social, biological, and financial networks. Graph learning consists loosely of learning functions from the set of graph isomorphism classes to the set of real numbers, and such functions include node classification, link prediction, and graph classification. Learning on graphs demands effective representation, usually in vector form, and different approaches include graph kernels \cite{kernelssurvey}, deep learning \cite{deepgraphsurvey}, and persistent homology \cite{persistentgraphsurvey}. Recently there has been a growing interest in understanding the discriminative power of certain frameworks \cite{powerful, understandinggnns, WLneural, Loukas2020What} which belongs to the inquiry into what functions on graph isomorphism classes can be learned. We call this the problem of function approximation on graphs. In machine learning, the problem of using neural networks (NNs) for function approximation on $\mathbb{R}^d$ is well-studied and the universal function approximation abilities of NNs as well as recurrent NNs (RNNs) are well known \cite{LESHNO1993861, SIEGELMANN1995132}. In this work, we propose a theoretical foundation for universal function approximation on graphs, and in Section \ref{sec:method} we present an algorithm with universal function approximation abilities on graphs. This paper will focus on the case of graph classification, but with minor modifications, our framework can be extended to other tasks of interest. We take care to develop a framework that is applicable to real-world graph learning problems and in Section \ref{sec:experiments} we show our framework performing at state-of-the-art on graph classification on four well known datasets and discriminating between graphs that other graph learning frameworks cannot.

%\bjn{State the problem first: what is a graph learning problem?  It is good to state the problems you want to solve (e.g. isomorphism) before getting into methods.  You should also define graph representation early on.}
Among deep learning approaches, a popular method is the graph neural network (GNN) \cite{GNNsurvey} which can be as discriminative as the Weisfeiler-Lehman graph isomorphism test \cite{powerful}. In addition, Long Short Term Memory models (LSTMs) that are prevalent in Natural Language Processing (NLP) have been used on graphs \cite{GraphRNN}. % where graphs are preprocessed with graph kernels (similar to word embeddings in NLP).
%and paths in the graphs are selected and processed by a LSTM to generate features of the graph.
%While results are good, these LSTM approaches are less end-to-end compared to GNNs, and, in general, comparisons of models that use node embeddings to those that do not have not been robust.
Using persistent homology features for graph classification \cite{GraphFiltration} also show promising results.
%With the growing use of persistent homology within machine learning \cite{topologylayer} it is worth investigating whether machine learning can be incorporated more deeply into the persistence framework. 
Our work borrows ideas from persistent homology \cite{persistentsurvey} and tree-LSTMs \cite{treelstm}.

% Our work is inspired by the filtration order that is used by persistent homology and partly explores if this order combined with tree-LSTMs [] can be used to improve on default persistent homology features. \bjn{I would motivate the use/inspiration of persistent homology more strongly, e.g. not "it is interesting to investigate..." but, "this work explores how deep learning techniques can be used to strengthen graph representations based on persistent homology"}
%For persistence over nodes and edges, nodes create connected components, while edges kill them. This computation requires a non-strict ordering of the edges and the sublevel filtration of a connected graph until a single connected component appears is equivalent to a minimal spanning tree of the graph. Thus, a tree-LSTM [] can be used to process the node-features and the connectivity of such a tree. However, the problem with this approach is that a minimum spanning tree, for many graphs, misses a lot of information since it might not process all edges. A related approach that is capable of capturing all connectivity simply sorts all edges and uses a tree-LSTM recursively to process them. This method creates features for the connected components that are constructed as more edges are processed, and ultimately we end up with a final set of features for the whole graph.

To be able to discriminate between any isomorphism classes, graph representation should be an injective function on such classes.
%Other things being equal, a graph representation should be an injective function on graph isomorphism classes to discriminate between any such classes. 
In practice this is challenging. 
%Any graph on $n$ vertices can be encoded by a string of length ${n \choose 2}$ that has $n!$ permutations in the same isomorphism class. Thus, a naive isomorphism-injective representation consists in encoding a graph into $n!$ strings of length ${n \choose 2}$. However, this is very computationally expensive and unlikely to be conducive to learning. 
%Unfortunately, 
Even the best known runtime \cite{babai} for such functions is too slow for most real world machine learning problems and their resulting representation is unlikely to be conducive to learning. To our knowledge, there exists no algorithm that produces isomorphism-injective graph representation for machine learning applications. 
We overcome several challenges by considering multivalued functions, with certain injective properties, on graph isomorphism classes instead of injective functions. 
%However, almost all our proofs are put in the Appendix.
%Naively encoding a graph isomorphism class $G$ into a string of length $n \choose 2$ is equiavelent to a multivalued function on graph isomorphism classes and with excellent runtime. Such an encoding can be used to discriminate between any graph isomorphism classes. We will use the approach of multivalued functions to reduce runtime but at the same time create flexible encodings that are conducive to learning in a way that a simple string encoding is not. Specifically, our method is structured to facilitate the detection of shared subgraph between graphs and to be as flexible as possible in its ability to learn a metric on the set of graphs. Ultimately, our method is a restricted universal function approximator on graph isomorphism classes, and it produces a collection of features that have desirable properties.  
% We build up basic theory for universal function approximation on graphs as well as present an algorithm with desirable universal approximation properties; however, almost all proofs are put in the Appendix. We also present some experiments and comparisons to gauge the usefulness of the algorithm.

Our main contributions: (i) Showing that graph representation with certain injective properties is sufficient for universal function approximation on bounded graphs and restricted universal function approximation on unbounded graphs. (ii) A novel algorithm for learning on graphs with universal function approximation properties, that allows for novel analysis, and that achieves state-of-the-art performance on four well known datasets. Our main results are stated and discussed in the main paper, while proof details are found in the Appendix.
	
\section{Theory}

An overview of this section:
(i) Multivalued functions, with injective properties, on graph isomorphism classes behave similarly to injective functions on the same domain. (ii) Such functions are sufficient for universal function approximation on bounded graphs, and (iii) for restricted universal function approximation on unbounded graphs. (iv) We postulate what representation of graphs that is conducive to learning. (v) We relate universal function approximation on graphs to the isomorphism problem, graph canonization, and discuss how basic knowledge about these problems affects the problem of applied universal function approximation on graphs. (vi) We present the outline of an algorithmic idea to address the above investigation.

% ; we explore how such approximation is restricted.
\subsection{Preliminaries}

\begin{figure}[h]
\centering
\begin{tikzpicture}
  \matrix (m) [matrix of math nodes,row sep=6em,column sep=24em,minimum width=2em]
  {
     \bm{\mathcal{G}}= \{[G], [H], ... \} & Y_1 \\
     \mathcal{G} = \{G_1,..., G_k, H_1, ..., H_l, ... \} & Y_2 \\};
  \path[-stealth]
    (m-1-1.east|-m-1-2) edge node [below] {\scriptsize{Injective: $f([G])=f([H]) \Rightarrow [G]=[H]$ (difficult)}}
            node [above] {$f$} (m-1-2)
    (m-2-1.east|-m-2-2) edge node [below] {\scriptsize{Iso-injective: $g(G)=g(H) \Rightarrow [G]=[H]$ (easy)}}
            node [above] {$g$} (m-2-2)
    (m-2-2) edge node [left] {$f \circ \bm{g}^{-1} $} (m-1-2)
    % (m-1-1) edge [double] node [above] { $\ \ \ \bm{g}$} (m-2-2);
    (m-2-1) edge node [above] { $(f \circ \bm{g}^{-1}) \circ g \ \ \ \ \ \ \ \ \ \ \ \ \ \ \ \ \ \ \ \ \ \ $} 
            node [below] {$\ \ \ \ \ \ \ \ \ \ \ \ \ =f \circ (/ \simeq)$} (m-1-2)
    (m-2-1) edge node [left] {$/\simeq$} (m-1-1);
\end{tikzpicture}
\caption{Diagram of the relations between injective functions on graph isomorphism classes, $\bm{\mathcal{G}}$, and iso-injective functions on graphs, $\mathcal{G}$. Constructing iso-injective functions on $\mathcal{G}$ is much easier than constructing injective functions on $\bm{\mathcal{G}}$, and by the existence of the well-defined function $f \circ \bm{g}^{-1}$ we do not lose much by switching our attention to iso-injective functions on $\mathcal{G}$. }
\label{fig:diagram}
\end{figure}
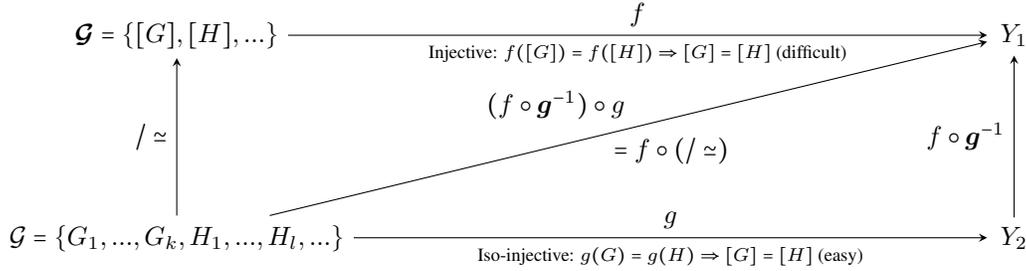

\begin{definition}
A \textit{graph} (undirected multigraph) $G$ is an ordered triple $G:=(V(G), E(G), l)$ with $V(G):=\{ 1,2,\dots,n\}$ a set of vertices or nodes, $E(G)$, a multiset of $m$ unordered pairs of nodes, called edges, and a label function $l:V(G) \rightarrow \mathbb{N}_+$ on its set of nodes. The \textit{size of graph} $G$ is $|G| := |V(G)|+|E(G)|+ \sup\{ l(v) \ | \ v \in V(G) \}$, and we assume all graphs are finite. % in terms of nodes, edges, and labels. 
\end{definition}

% \begin{definition}
% A \textit{subgraph} $S$ of a graph $G$, denoted $S \subset G$, is another graph formed from a subset of the vertices and edges of $G$. The vertex subset must include all endpoints of the edge subset, but may also include additional vertices. 
% \end{definition}

\begin{definition}
Two graphs $G$ and $H$ are \textit{isomorphic} ($G \simeq H$) if there exists a bijection $\phi:V(G) \rightarrow V(H)$ that preserves edges and labels, i.e. a graph \textit{isomorphism}.
\end{definition}

% \begin{definition}
% For a graph $G$ and any subgraphs $S_1,S_2 \subset G$, $S_1 \equiv S_2$ implies that $S_1$ and $S_2$ have the same set of nodes, $V(S_1)=V(S_2)$, and the same multiset of edges, $E(S_1)=E(S_2)$.
% \end{definition}

% Our algorithms use actual instantiations of graphs, and if two graphs $G$ and $H$ refers to the same instantiation we write $G \equiv H$ (this is stronger than the existence of an isomorphism). Similarly, for binary set operations on $G$ and $H$ such as $V(G) \cap V(H)$ we mean in terms of the instantiation of $G$ and $H$. We only use these on pairs of subgraphs $S_1,S_2$ of a graph $G$, and then they can be thought to be in terms of the indexing of the nodes of $G$ which then corresponds to its instantiation.

%%% REDUNDANT?
% \begin{remark}
% Functions on nodes $f:V(G) \rightarrow Y$, such as node labels, are functions of graphs too, because it makes no sense to compare indices or nodes between different graphs that are not subgraphs of the same graph. That is, each such function is different for each graph $G$, so if we abuse notation when having also a graph $H$ and $f:V(H) \rightarrow Y$ in a shared context with $G$, then $v_1=v_2$ implies $f(v_1)=f(v_2)$ only if $v_1,v_2 \in V(G)$ or $v_1,v_2 \in V(H)$.
% \end{remark}

\begin{definition} Let $\mathcal{G}$ denote the set of all finite graphs. For $b \in \mathbb{N}$ let $\mathcal{G}_b \subset \mathcal{G}$ denote the set of graphs whose size is bounded by $b$.
\end{definition}

\begin{definition}
Let $\bm{\mathcal{G}}$ denote the set of all finite graph isomorphism classes, i.e. the quotient space $\mathcal{G}/\simeq$. For $b \in \mathbb{N}$ let $\bm{\mathcal{G}}_b \subset \bm{\mathcal{G}}$ denote the set of graph isomorphism classes whose size is bounded by $b$, i.e. $\mathcal{G}_b/\simeq$. In addition, we denote the graph isomorphism class of a graph $G \in \mathcal{G}$ as $[G]$ (coset) meaning for any graphs $G,H \in \mathcal{G}$, $[G]=[H]$ if and only if $G \simeq H$.
\end{definition}

\begin{lemma}
\label{lemma:graphsetscardinality}
The sets $\mathcal{G}$ and $\bm{\mathcal{G}}$ are countably infinite, and the sets $\mathcal{G}_b$ and $\bm{\mathcal{G}}_b$ are finite.
\end{lemma}

% \begin{definition}
% If we write $f(A)$ where $A$ is a subset of the domain of $f$, we mean the multiset $f(A) := \{f(x) \ | \ x\in A \}_M$ (the subscript $M$ denotes multiset construction).
% \end{definition}

\begin{definition}
A function $f:\mathcal{G}\rightarrow Y$ is \textit{iso-injective} if it is injective with respect to graph isomorphism classes $\bm{\mathcal{G}}$, i.e. for $G,H \in \mathcal{G}$, $f(G)=f(H)$, implies $G \simeq H$.
\end{definition}

% \begin{proofatend}
% Suppose the isomorphism class $\bm{G}$ have $|V(\bm{G})|+|E(\bm{G})|+\max_{v\in V(\bm{G})}(l(v)) = n$, there are only finitely many such graphs $G$.
% \end{proofatend}

\begin{definition}
A \textit{multivalued function} $f:X \Rightarrow Y$ is a function $f:X \rightarrow \mathcal{P}(Y)$, i.e. from $X$ to the powerset of $Y$, such that $f(x)$ is non-empty for every $x \in X$. 
\end{definition}

\begin{definition}
Any function $f:\mathcal{G}\rightarrow Y$ can be seen as a multivalued function $\bm{f}:\bm{\mathcal{G}} \Rightarrow Y$ defined as
$ \bm{f}([G]) := \{ f(H) \ | \ H \in [G] \} $
and we call the size of the set $\bm{f}([G])$ the \textit{class-redundancy} of graph isomorphism class $[G]$.
%\bjn{$\{ f(H) \mid H\in \bm{G}\}$ ???}
\end{definition}

Let $Alg: \mathcal{G} \rightarrow \mathbb{R}^d$ be an iso-injective function. For a graph $G \in \mathcal{G}$ we call the output of $Alg(G)$ the \textit{encoding} of graph $G$. The idea is to construct a universal function approximator by using the universal function approximation properties of NNs. We achieve this by composing $Alg$ with NNs and constructing $Alg$ itself using NNs. Without something similar to an injective function $f:\bm{\mathcal{G}}\rightarrow Y$ we will not arrive at a universal function approximator on $\bm{\mathcal{G}}$. However, we do not lose much by using a multivalued function $\bm{g}:\bm{\mathcal{G}} \Rightarrow Y$ that corresponds to an iso-injective function $g:\mathcal{G} \rightarrow Y$.

\begin{theorem}
\label{thm:fromisotoinjective}
For any injective function $f:\bm{\mathcal{G}}\rightarrow Y$ and iso-injective function $g:\mathcal{G} \rightarrow Y$ there is a well-defined function $h:\operatorname{im}(g) \rightarrow Y$ such that $f = h \circ g$.
\end{theorem}

% \begin{figure}
%     \centering
%     \begin{align*}
%     \bm{\mathcal{G}}= \{[G], [H], ... \} &\xrightarrow[\text{(injective)}]{f} \mathbb{R}^d \text{ s.t. } f([G]) = f([H]) \implies [G] = [H] \text{ is hard, but } \\
%     \mathcal{G} = \{G_1,..., G_k, H_1, ..., H_l, ... \} &\xrightarrow[\text{(iso-injective)}]{g} \mathbb{R}^d \text{ s.t. } g(G) = f(H) \implies [G] = [H] \text{ is easy } \\
%     \text{ and } f=(f \circ \bm{g}^{-1})& \circ g \text{ where } f \circ \bm{g}^{-1} : \mathbb{R}^d \rightarrow \mathbb{R}^d \text{ is a well-defined function.}
%     \end{align*}
%     \caption{Injective functions on graph isomorphism classes versus iso-injective functions on graphs.}
%     \label{fig:inj-VS-iso}
% \end{figure}

See Figure \ref{fig:diagram} for a diagram relating these different concepts. For completeness, we also add the following theorem.

\begin{theorem}[recurrent universal approximation theorem \cite{SIEGELMANN1995132}]
\label{thm:recurrentuniversalapproximationtheorem}
For any recursively computable function $f:\{0,1\}^* \rightarrow \{0,1\}^*$ there is a RNN $\phi$ that computes $f$ with a certain runtime $r(|w|)$ where $w$ is the input sequence.
\label{thm:recur}
\end{theorem}

%\begin{remark}
Unfortunately Theorem \ref{thm:recur} requires a variable number of recurrent applications that is a function of the input length, which can be hard to allow or control. Furthermore, the sets of graphs we analyze are countable. This makes for a special situation, since a lot of previous work focuses on NNs' ability to approximate Lebesgue integrable functions, but countable subsets of $\mathbb{R}$ have measure zero, rendering such results uninformative. Thus, we focus on pointwise convergence. 
%\end{remark}

\subsection{Bounded Graphs}

With an iso-injective function, universal function approximation on bounded graphs is straightforward.

\begin{theorem}[finite universal approximation theorem]
For any continuous function $f$ on a finite subset $X$ of $\mathbb{R}^d$, there is a NN $\varphi$ with a finite number of hidden layers containing a finite number $n$ of neurons that under mild assumptions on the activation function can approximate $f$ perfectly, i.e.
$ ||f - \varphi||_{\infty} = \sup_{x \in X}|f(x)-\varphi(x)| = 0 $.
\label{thm:finiteuniversal}
\end{theorem}

From Theorem \ref{thm:fromisotoinjective} and since $\mathcal{G}_b$ is finite we arrive at the following:

\begin{theorem}
\label{thm:finiteboundediso}
Any function $f: \bm{\mathcal{G}}_b \rightarrow \mathbb{R}$ can be perfectly approximated by any iso-injective function $Alg:\mathcal{G}_b \rightarrow \mathbb{R}^d$ composed with a NN $\varphi:\mathbb{R}^d \rightarrow \mathbb{R}$.
\end{theorem}

%%%%%%%%%% UNBOUNDED GRAPHS
\subsection{Unbounded Graphs}

For a function to be pointwise approximated by a NN, boundedness of the function and its domain is essential. Indeed, in the Appendix we prove (i) there is no finite NN with bounded or piecewise-linear activation function that can pointwise approximate an unbounded continuous function on an open bounded domain, and (ii) there is no finite NN with an activation function $\sigma$ and $k \geq 0$ such that $\frac{d^k \sigma}{d x^k}=0$ that can pointwise approximate all continuous functions on unbounded domains.

% \begin{theorem}
% \label{thm:noboundedplinear}
% There is no finite width and depth NN with bounded or piecewise-linear activation function that can pointwise approximate an unbounded continuous function on an open bounded domain.
% \end{theorem}

% \begin{theorem}
% \label{thm:nofinitenoneinifintediff}
% There is no finite width and depth NN with an activation function $\sigma$ and $k \geq 0$ such that $\frac{d^k \sigma}{d x^k}=0$ that can pointwise approximate all continuous functions on unbounded domains.
% \end{theorem}

\begin{theorem}[universal approximation theorem \cite{LESHNO1993861}]
For any $\epsilon > 0$ and continuous function $f$ on a compact subset $X$ of $\mathbb{R}^d$ there is a NN $\varphi$ with a single hidden layer containing a finite number $n$ of neurons that under mild assumptions on the activation function can approximate $f$, i.e.
$ ||f - \varphi||_{\infty} = \sup_{x \in X}|f(x)-\varphi(x)| < \epsilon $.
\label{thm:universal}
\end{theorem}

Though universal approximation theorems come in different forms, we use Theorem \ref{thm:universal} as a ballpark of what NNs are capable off. As shown above, continuity and boundedness of functions are prerequisites. This forces us to take into account the topology of graphs. Indeed, any function $f:\bm{\mathcal{G}} \rightarrow \mathbb{R}^d$ with a \textit{bounded} co-domain will have a convergent subsequence for each sequence in $\bm{\mathcal{G}}$, by Bolzano-Weierstrass. Since a NN $\varphi:\mathbb{R}^d \rightarrow \mathbb{R}^d$ may only approximate continuous functions on $\operatorname{im}(f)$, the same subsequences will be convergent under $\varphi \circ f$. Thus, since $\bm{\mathcal{G}}$ is countably infinite and due to limiting function approximation abilities of NNs, we always, for any $f$, have a convergent \textit{infinite} sequence \textit{without repetition} of graph isomorphism classes. Furthermore, $f$ \textit{determines} such convergent sequences independent of $\varphi$ and should therefore be learnable and flexible so that the convergent sequences can be adapted to the specific task at hand.
%; in fact, for an iso-injective function $Alg:\mathcal{G} \rightarrow \mathbb{R}^d$, since for every graph $G$, $[G]$ is finite, a sequence $G_i$ is convergent in $\mathcal{G}$ if and only if a corresponding sequence $[G_i]$ is convergent in $\bm{\mathcal{G}}$.
See Appendix for more details. This leads to the following remark: \\

\begin{remark}
An injective function $f:\bm{\mathcal{G}} \rightarrow \mathbb{R}^d$ determines a non-empty set of convergent \textit{infinite} sequences \textit{without repetition} in $ \bm{\mathcal{G}} $ under the composition $g =\varphi \circ f$ with any NN $\varphi$. Meaning that $f$ affects which functions $g$ can approximate. Thus, for flexible learning, $f$ should be flexible and learnable to maximize the set of functions that can be approximated by $g$. Hopefully then, we can learn an $f$ such that two graphs $[G]$ and $[H]$ that are close in $||f([G])-f([H])||$ are also close according to some useful metric on $\bm{\mathcal{G}}$. The same holds for iso-injective functions $Alg:\mathcal{G} \rightarrow \mathbb{R}^d$.
%and for any $f$ there is a corresponding $Alg$ such that $Alg|_{\bm{\mathcal{G}}}=f$ and a sequence $G_i$ is convergent for $Alg$ if and only if a corresponding sequence $[G_i]$ is convergent for $f$. I.e. the topoloy on $\bm{\mathcal{G}}$  
\label{remark:converge}
\end{remark}

We are left to create a function $Alg: \mathcal{G} \rightarrow \mathbb{R}^d$ that is bounded but we cannot guarantee it will be closed so that we may use Theorem \ref{thm:universal}; however, we add this tweak:

\begin{theorem}
\label{thm:2ndapproxthm}
For any $\epsilon > 0$ and bounded continuous function $f$ on a bounded subset $X$ of $\mathbb{R}^d$ there is a NN $\varphi$ with a single hidden layer containing a finite number $n$ of neurons that under mild assumptions on the activation function can approximate $f$, i.e.
$ ||f - \varphi||_{\infty} = \sup_{x \in X}|f(x)-\varphi(x)| < \epsilon $.
\end{theorem}

For example, we can bound any iso-injective function $Alg: \mathcal{G} \rightarrow \mathbb{R}^d$ by composing (this simply forces the convergent sequences to be the values in $\mathbb{R}^d$ with increasing norm) with the injective and continuous Sigmoid function $ \sigma(x) = \frac{1}{1+e^{-x}} $.

% We will show how our framework can approximate such a function, that is both bounded and flexible in approximating what graphs are mapped close to each other. Next, we will look at graphs of bounded size.

\subsection{Learning and Graph Isomorphism Problems}

\begin{definition}
The \textit{graph isomorphism problem} consists in determining whether two finite graphs are isomorphic, and \textit{graph canonization} consists in finding, for graph $G$, a canonical form $Can(G)$, such that every graph that is isomorphic to $G$ has the same canonical form as $G$. 
\end{definition}

% \begin{remark}
% In essence, a canonical form $Can:\mathcal{G} \rightarrow Y$ is an injective function on graph isomorphism classes. Thus, from a solution to the graph canonization problem, one could also solve the problem of graph isomorphism.
% \end{remark}

% \begin{remark}
% For an iso-injective function $f:\mathcal{G} \rightarrow Y$, the set $\bm{f}([G])$ of a graph $G$ is a canonical form.
% \end{remark}

The universal approximation theorems say nothing about the ability to learn functions through gradient descent or generalize to unseen data. Furthermore, a class of graphs occurring in a learning task likely contains non-isomorphic graphs. Therefore, to direct our efforts, we need a hypothesis about what makes learning on graphs tractable. 
%We postulate the following:

\begin{postulate}
A representation (encoding) of graphs that facilitates the detection of shared subgraphs (motifs) between graphs is conducive to learning functions on graphs.
\end{postulate}

With this in mind, an ideal algorithm produces for each graph a representation consisting of the multi-set of canonical forms for all subgraphs of the graph. Even better if the canonical representations of each graph are close (for some useful metric) if they share many isomorphic subgraphs. However, there is a few challenges: (i) The fastest known algorithm for the graph canonization problem runs in quasipolynomial $2^{O((\log n)^c)}$ time \cite{babai}, and (ii) a graph has exponentially $\Omega(n!)$ many distinct subgraphs. 
% \end{enumerate}

% \begin{enumerate}
%     \item The fastest known algorithm for the graph canonization problem runs in quasipolynomial $2^{O((\log n)^c)}$ time.
%     \item A graph has exponentially $O(n!)$ many distinct subgraphs. 
% \end{enumerate}

First, obtaining a canonical form of a graph is expensive and there is no guarantee that two graphs with many shared subgraphs will be close in this representation. Second, obtaining a canonical form for each subgraph of a graph is even more ungainly. We approach these challenges by only producing iso-injective encodings of a graph and a sample of its subgraphs.
% \begin{definition} % (Good for learning)
% We say a \textit{multi-set of encodings} for a graph is \textit{iso-injective} if each feature of the multi-set corresponds to a subgraph of the graph and each encoding is \textit{iso-injective}.
% \end{definition}
Iso-injective encodings of graphs are easily obtained in polynomial time. However, we still want small class-redundancy and flexibility in learning the encodings.

\subsection{Algorithmic Idea}

We construct a \textit{universal function approximator on graph isomorphism classes of finite size} by constructing a multi-set of encodings that are \textit{iso-injective}. Ideally, for efficiency, an algorithm when run on a graph $G$ constructs iso-injective encodings for subgraphs of $G$ as a subprocess in its construction of an iso-injective encoding of $G$. Thus, a recursive local-to-global algorithm is a promising candidate. Consider Algorithm \ref{alg:outline2}; the essence of subset parsing is the following:

% Thus, an algorithm that formulate the graph isomorphism problem as a search problem or a Constraint Satisfaction Problem is not bound to be very helpful for learning. 

% Now, we can preprocess a training set of graphs to get a representation for learning, but ideally we still want the algorithm that produces the representation to be fast during evaluation on unseen graphs. 

%(Differentiability?)

\begin{algorithm}[h]
\caption{Subset Parsing Algorithm}
\label{alg:outline2}
%\SetAlgoLined
\begin{algorithmic}
\STATE {\bfseries Input:} Graph $G$, \\
% \STATE Let $G(V(G),E(G),l)$ be a undirected graph with vertices $V(G)=\{v_1,\dots,v_n\}$, edges $E(G)=\{e_1,\dots,e_m\}$, and label function $l:V(G) \rightarrow \mathbb{N}_+$\;
%  $x = X_{\bm{G}}(\omega_t)$ is the evaluation of the random variable $X_{\bm{G}}$ where $\omega_t$ is the sample of this run on graph isomorphism class $\bm{G}$\; 
\hspace{0.25cm} set $A$ of subgraphs of $G$, and  functions $c:A \rightarrow \mathbb{R}^{d_c}$, \ \ $r:\{\mathbb{R}^{d_c}, \mathbb{R}^{d_c}\}\times \mathcal{P}(h(V)) \times \mathbb{N} \rightarrow \mathbb{R}^{d_c}$
%\hspace{0.25cm} functions $c:A \rightarrow \mathbb{R}^{d_c}$,
%\STATE Let $h:V \rightarrow \mathbb{R}^{d_h}$ be a label function on the nodes\;
%\ \ $r:\{\mathbb{R}^{d_c}, \mathbb{R}^{d_c}\}\times \mathcal{P}(h(V)) \times \mathbb{N} \rightarrow \mathbb{R}^{d_c}$
\STATE {\bfseries Output:} Extended function $c:A \rightarrow \mathbb{R}^{d_c}$
\FOR{$S_1, S_2 \in A$}
    \STATE Let $S_{1,2} = S_1 \cup S_2$\;
    \STATE $c(S_{1,2}) = r(\{c(S_1), c(S_2)\}, \{ l(v) \ | \ v \in V(S_1) \cap V(S_2) \}, 
    |V(S_{1,2})|+|E(S_{1,2})|)$\; % already encoded edges in intersection
   \STATE $A = A \cup \{S_{1,2}\}$ %(A - \{S_1,S_2\}) \cup \{S_{1,2}\} $\;
\ENDFOR

 \end{algorithmic}
\end{algorithm}

%\bjn{explicitly say what the algorithm returns}

\begin{theorem}
For Algorithm \ref{alg:outline2} the encoding $c(S_{1,2})$ with $S_{1,2} = S_1 \cup S_2$ and $|V(S_{1,2})|+|E(S_{1,2})|=p>1$ is iso-injective if we have on input graph $G$
\begin{enumerate}
  \item for all $S \in A \subset G$, with $|V(S)|+|E(S)| < p$ %, with $p > 1$
  \begin{enumerate}
      \item the encoding $c(S)$ is iso-injective
      \item each label $l(v)$ for $v \in V(S)$ is unique 
    \end{enumerate}
  \item $r$ is an injective function
\end{enumerate}
\label{thm:outline2}
\end{theorem}

%\begin{remark}
% If we consider making this into a implementable algorithm, we can injectively encode input $h(V(s_1) \cap V(s_2))$ to $r$ as a $\sum_{s \in V(s_1) \cap V(s_2)}  f(h(s))$ (such a function exists if the size of the intersection bounded). However, if we know the intersection is bounded, we don't need the number of nodes to be bounded. We also need to make sure that all $h$-values in $s_{1,2}$ are unique after we go onto the next step. 

We envision an algorithm that combines encodings of subgraphs $S_{1}, \dots, S_{n}$ into an encoding of graph $S_{1,\dots, n}$, such that if $c(S_{1}), \dots, c(S_{n})$ are iso-injective so is $c(S_{1,\dots, n})$. However, we need to make sure all labels are unique within each subgraph and to injectively encode pairwise intersections. 
%We present one viable approach.

%  Arbitrary intersection is not a problem per se, but handling it properly appears complicated, at least on the surface. At the same time, this reveals that there are ways of using subset parsing to make GNNs extract iso-injective features.
 
%  Instead, we will construct an algorithm that combines two subgraphs an edge at a time, since as we will see, this will simplify the encoding of intersection as well as making sure that $h$-values are unique within each processed subgraph.

%%%%%%%%%%%%%
% METHOD
%%%%%%%%%%%%%
\section{Method}
\label{sec:method}

Methods such as GNNs successfully aggregate label and edge information in a local-to-global fashion; however, GNNs lack sufficiently unique node identification to extract fully expressive representations \cite{powerful}. The quickly growing number (unbounded for graphs in $\mathcal{G}$) of intersections in GNNs' processing of subgraphs complicates analysis. Our method keeps processed subgraphs disjoint (Lemma \ref{lemma:disjoint}) which allows for comparatively simple inductional analysis. We ensure that within a processed subgraph each node-encoding is unique, which together with some additional properties proves sufficient to produce iso-injective encodings for graphs (Theorem \ref{thm:NPAtheorem}). Parsing disjoint subgraphs by adding one edge at a time is inspired by $0$-dimensional persistent homology \cite{persistentsurvey}; the idea being that our method may revert to computing 0-dimensional persistence based on increasing node degrees, and should therefore (neglecting overfitting) perform no worse than certain persistence based kernels \cite{persistentgraphsurvey, GraphFiltration}. See Figure \ref{fig:message_passing} for how message (or information) passing occurs in Node Parsing (Algorithm \ref{alg:NPA}) versus in GNNs.
\begin{figure}[h]
\includegraphics[width=\textwidth]{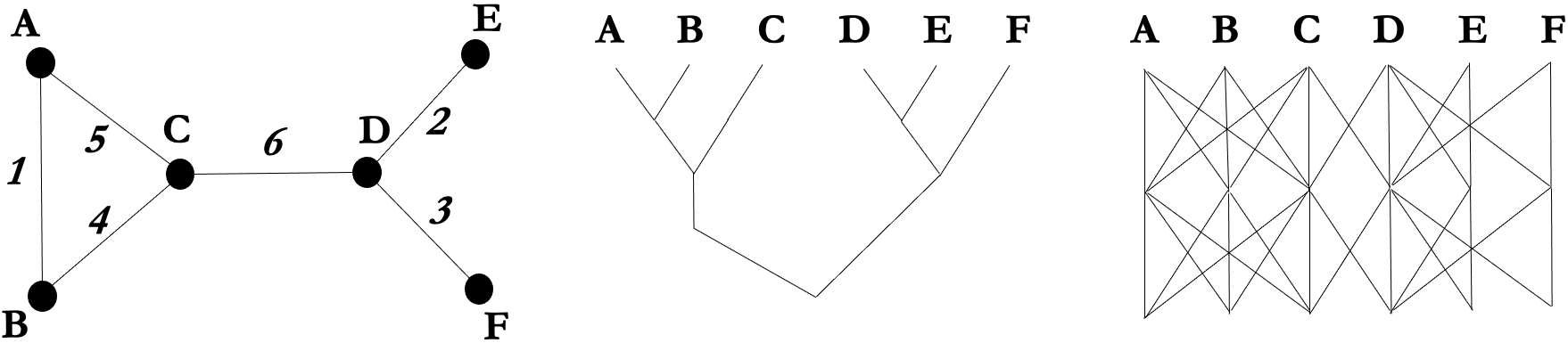}
\caption{Left to right: Graph with edge-ordering. Message passing in Node Parsing on graph. Message passing in GNN on same graph. }
\label{fig:message_passing}
\end{figure}

%For subset parsing methods it is quite easy to construct counter examples if the encoding of nodes is not unique within a processed subgraph. 

In this section we present Algorithm \ref{alg:NPA} and show how with the use of NNs it is a \textit{universal function approximator on graphs} (Theorem \ref{thm:NPAUniversal}). 
%We also present additional analysis as well as the weaker baseline Algorithm \ref{alg:NPBA} for comparison. 
This section is outlined as follows: (i) A description of the Node Parsing Algorithm (NPA). (ii) Proving that, under certain requirements on the functions that NPA make use of, NPA produces iso-injective representations of graphs. (iii) Proving the existence of functions with the prerequisite requirements. (iv) Proving NNs can approximate such functions. (v) Presenting a weaker baseline model for comparison. (vi) Analysis of class-redundancy, parallelizability, and introducing the concept of subgraph droupout.

\subsection{The Algorithm}

% \begin{algorithm}[tb]
%   \caption{Bubble Sort}
%   \label{alg:example}
% \begin{algorithmic}
%   \STATE {\bfseries Input:} data $x_i$, size $m$
%   \REPEAT
%   \STATE Initialize $noChange = true$.
%   \FOR{$i=1$ {\bfseries to} $m-1$}
%   \IF{$x_i > x_{i+1}$}
%   \STATE Swap $x_i$ and $x_{i+1}$
%   \STATE $noChange = false$
%   \ENDIF
%   \ENDFOR
%   \UNTIL{$noChange$ is $true$}
% \end{algorithmic}
% \end{algorithm}

\begin{algorithm}[h]
\caption{Node Parsing Algorithm (NPA)}
\label{alg:NPA}
%\SetAlgoLined
% \KwResult{ A sequence of feature vectors $w_1,...,w_n,w_{n+1},...,w_{n+m}$. }

\begin{algorithmic}
\STATE {\bfseries Input:} Graph $G$,
%\STATE Let $G(V(G),E(G), l) \in \mathcal{G}$ be a undirected graph with vertices $V(G)=\{v_1,\dots,v_n\}$, edges $E(G)=\{e_1,\dots,e_m\}$, and label function $l:V(G) \rightarrow \mathbb{N}_+$\; 
%  $x = X_{\bm{G}}(\omega_t)$ is the evaluation of a random variable $X_{\bm{G}}$ and where $\omega_t$ represents this run of the algorithm on graph isomorphism class $\bm{G}$\;
%  For each $v \in V$, let $l_{v}$ be the label vector associated with each vertex\; \\
%\STATE A label function $l:V(G) \rightarrow \mathbb{N}_+$\; 
\STATE \hspace{0.25cm} functions $s_e:E \times \mathcal{G} \rightarrow \mathbb{R}$, \ \ $s_v:V \times \mathcal{G} \rightarrow \mathbb{R}$, \ \ $h_{init}:\mathbb{N}_+ \rightarrow \mathbb{R}^{d_v}$, \ \ $c_{init}:\mathbb{R}^{d_v} \rightarrow \mathbb{R}^{d_c-1}$
%\STATE \hspace{0.25cm} functions $h_{init}:\mathbb{N}_+ \rightarrow \mathbb{R}^{d_v}$, \ \ $c_{init}:\mathbb{R}^{d_v} \rightarrow \mathbb{R}^{d_c-1}$, 
\STATE \hspace{0.25cm} functions $r_c: \mathbb{R}^{2d_c+2d_v}\times \{0,1\} \rightarrow \mathbb{R}^{d_c}$, \ \ $r_{v}: \mathbb{R}^{d_c+d_v} \times \{0,1\} \rightarrow \mathbb{R}^{d_v}$,
\STATE \hspace{0.25cm} special symbol $zero$
% \STATE {\bfseries Output:} Multisets $W(G):=\{w_i \ | \ i = 1,\dots,n+m \}_M$, \ $C(G):=\{c(S) \ | \ S \in A_{m+1} \}_M \subset W(G)$
\STATE {\bfseries Output:} Multisets $W(G):=[w_i \ | \ i = 1,\dots,n+m ]$, \ $C(G):=[c(S) \ | \ S \in A_{m+1} ] \subset W(G)$
%, set $A_{m+1}$, map $c:A_{m+1}\rightarrow \mathbb{R}^{h_c}$
\STATE Let $A_1 =  V(G)$\; // Where each node is seen as a subgraph of $G$ %\{ v \ | \ v \in V(G) \}$\;
\FOR{$i=1, \dots, n$}
\STATE $h^1(v_i) = h_{init}(l(v_i))$\; 
\STATE $w_{i} = c(v_i) = c_{init}(h^1(v_i))$.append($zero$) // step $0$ encode \;
%\STATE $w_{i} = c(v_i)$\; 
\ENDFOR
\STATE Sort $E$ with $s_e(\cdot,G)$ so $s_e(e_1,G), \dots, s_e(e_m,G)$ are in ascending order\; 
 
\FOR{$i=1, \dots, m$}
\STATE Let $(v_a,v_b) = e_i$ and sort $(v_a,v_b)$ ascendingly with $s_v(\cdot,G)$\;  
%\STATE Sort $(v_a,v_b)$ ascendingly with $s_v(\cdot,G)$
\STATE Let $S_1, S_2 \in A_i$ be subgraphs with $v_a \in S_1$ and $v_b \in S_2$\;
%\STATE Let $S_2 \in A_i$ be subgraph with $v_b \in S_2$\; 
\STATE Let $S_{1,2} = S_1 \cup S_2 \cup (v_a,v_b)$\;
\STATE $w_{n+i} = c(S_{1,2}) = r_c(\{ (c(S_1), h^{i}(v_a)), (c(S_2), h^{i}(v_b)) \}, \mathbbm{1}_{S_1 = S_2})$ // step $i$ encoding of $S_{1,2}$\;
\STATE $h^{i+1} = h^{i}$ // inheriting previous $h$-values\; 
 
\FOR{$v \in V(S_{1,2})$}
\STATE $h^{i+1}(v) = r_{v}(c(S_{1,2}), h^{i}(v), \mathbbm{1}_{v \in V(S_1)})$;\
\ENDFOR
  
\STATE $A_{i+1} = (A_{i} - \{S_1,S_2\}) \cup \{S_{1,2}\} $\;
  
%\STATE $w_{n+i} = c(S_{1,2})$\;
\ENDFOR

\end{algorithmic}
\end{algorithm}

%\bjn{explicitly say what algorithm returns} \bjn{there's a lot going on in this algorithm block.  Is there a way to break it up or simplify?}
%Note how in NPA (Algorithm \ref{alg:NPA}) the specific instantiation of a graph $G$ can affect the output of the algorithm via $s_e$ and $s_v$. This reflects the uncertainty of the actual representation (ordering of the nodes) of the graph isomorphism class $[G]$ as it is fed to our algorithm in the form of $G$.
%, e.g. the indexing of the nodes which we do not control. For the version of NPA we present, the number of possible encodings for a graph isomorphism class $[G]$ is typically much less than all possible orderings of its nodes. Either way, 
%However, for any finite graph isomorphism class $[G]$ the number of orderings of the nodes are finite, so the number of possible outputs is finite. This means that for any graph $G \in \mathcal{G}$ we can always compute the set $\bm{Alg}([G])$ in a finite number of steps.

% \begin{definition}
% For graph $G$ in $\mathcal{G}$ we define the multiset $W(G):=\{w_i \ | \ i = 1,\dots,n+m \}$ after running NPA on $G$.
% \end{definition}

%We want to prove that there exist functions under which NPA constructs iso-injective representation. First, we equip ourselves with a lemma:
\begin{lemma}
In Algorithm \ref{alg:NPA}, an edge (in the second for loop) is always between two disjoint subgraphs in $A_i$ or within the same (with respect to $=$) subgraph in $A_i$. Also, each subgraph in $A_i$ is disjoint and connected. \\ 
\label{lemma:disjoint}
\end{lemma} 

% \begin{remark}
% NPA produces a sequence of encodings for a graph $G$ but when finished, set $A_{m+1}$ contains each of the largest (by inclusion) disjoint connected subgraphs of $G$. Since NPA builds encodings recursively from disjoint subgraphs, NPA constructs encodings for each such largest subgraph independently as if it is run once for each of them. Thus, proving that NPA produces iso-injective encodings for connected graphs, implies each multiset $W(G)$ and $C(G)$ is iso-injective also for disconnected graphs. 
% \end{remark}

%Now we are ready to prove an important theorem:

\begin{theorem}
For Algorithm \ref{alg:NPA}, each produced $c$-encoding is iso-injective, if $h_{init}$, $c_{init}$, and $r_c$ are injective, if for all subgraphs $S_1, S_2 \in A_{i}$ that appear at step $i$ when run on input graph $G$
%in the algorithm we have:
\begin{itemize}
    \item each value $r_v(c(S_{1,2}), \tilde{h}, \mathbbm{1}_{v \in V(S_1)})$ for $\tilde{h} \in  h^{i}(V(S_1) \cup V(S_2))$ is unique,
\end{itemize}
and if for all graphs $S_{1,2}, S^*_{1,2}$ with $c:=c(S_{1,2})=c(S^{*}_{1,2})$, encoded at step $i$ run $G$ and step $j$ run $H$ respectively, 
%of the algorithm we have:
\begin{itemize}
    \item $r_{v}(c, \cdot, \mathbbm{1}_{v \in V(S_1)})$ is injective across $\{ h^{i}(v) \ | \ v \in V(S_{1,2}) \}$ and $\{ h^{j}(v) \ | \ v \in V(S_{1,2}^*) \}$ 
\end{itemize}
% \begin{itemize}
%     \item $r_{v}(c, \cdot, \mathbbm{1}_{v \in V(S_1)})$ is injective across domains $\operatorname{im}(h^{i}|_{V(S_{1,2})})$ and $\operatorname{im}(h^{j}|_{V(S_{1,2}^*)})$ 
% \end{itemize}
% then each produced $c$-value is an iso-injective encoding for graph isomorphism classes in the sense that if two connected graphs $G_1$ and $G_2$ have $c(G_1)=c(G_2)$ then $G_1 \simeq G_2$. 

% Consider Algorithm \ref{alg:NPA}. If $h_{init}$, $c_{init}$, and $r_c$ are injective and if for all subgraphs $S_1, S_2 \in A_{i}$ that appear at step $i$ on run $G$ in the algorithm, we have:
% \begin{itemize}
%     \item each value in the image $r_v(c(S_{1,2}), h^{i}(V(S_1) \cup V(S_2)), \mathbbm{1}_{v \in V(S_1)})$ is unique,
% \end{itemize}
% and for all graphs $S_{1,2}, S^*_{1,2}$ with $c=c_i=c(S_{1,2})=c^*_j=c(S^{*}_{1,2})$ that are encoded at step $i$ run $G$ and step $j$ run $H$, respectively, of the algorithm we have:
% \begin{itemize}
%     \item $r_{v}(c, \cdot, \mathbbm{1}_{v \in V(S_1)})$ is injective across domains $h^{i}(V(S))$ and $h^{j}(V(S^*))$ 
% \end{itemize}
% then each produced $c$-value is an iso-injective encoding for graph isomorphism classes in the sense that if two connected graphs $G_1$ and $G_2$ have $c(G_1)=c(G_2)$ then $G_1 \simeq G_2$. 
\label{thm:main}
\end{theorem}

% \begin{definition}
% With a function $f:X \rightarrow Y$ being injective across domains $X_1$ and $X_2$ with $X_1, X_2 \subset X$, we mean that for all $x_1 \in X_1, x_2 \in X_2$ with $f(x_1)=f(x_2)$ we have $x_1=x_2$.
% \end{definition}

By Lemma \ref{lemma:disjoint}, intersection is encoded by $\mathbbm{1}_{S_1 = S_2}$ and uniqueness of $h$-values is established by properties of $r_{v}$ (specifically, $\mathbbm{1}_{S_1 = S_2}$ allows us to discern whether a new edge is between two disjoint isomorphic subgraphs, with identical $c$-encodings, or within the same subgraph). Thus, the proof follows almost immediately from Theorem \ref{thm:outline2}. Furthermore, and critically, $r_v(c(S_{1,2}), \cdot, \mathbbm{1}_{v \in V(S_1)})$ being injective across $\{ h^{i}(v) \ | \ v \in V(S_{1,2}) \}$ and $\{ h^{j}(v) \ | \ v \in V(S_{1,2}^*) \}$ ensures that if we find that two graphs are isomorphic after having applied $r_v$ they were also isomorphic before the application of $r_v$, all the way back to the original node-labels. The special $zero$-symbol allows us to assert whether an encoded graph has zero edges, as we otherwise want to deconstruct an encoded subgraph by considering two earlier encoded subgraphs connected by an edge. 

% \begin{remark}
% From the proof we can also conclude that we do not need to update $h$-values when an edge is within the same connected component because then all the inductive hypotheses are already achieved.
% \end{remark}

\subsection{Existence of Required Functions}
\label{sec:reqfuncs}

%In providing functions with the prerequisite properties we rely on the fact that $\mathcal{G}$ is countably infinite, i.e. there is a bijection between $\mathcal{G}$ and $\mathbb{N}_0$. 
In providing functions with the prerequisite properties we rely on the fact that our labels live in $\mathbb{N}_+$. This is necessary since we want to be able to use NNs, which can only approximate continuous functions, while at the same time our method injectively compresses label and connectivity information. In particular, there exists a continuous and bounded function from $\mathbb{R}^2$ to $\mathbb{R}$ that is injective in $\mathbb{N}^2$, while there exists no continuous function from $\mathbb{R}^2$ to $\mathbb{R}$ that is injective in $\mathbb{R}^2$. %See Appendix for details.

Suppose the $c$-encoding of a subgraph $S_k$ consists of $c(S_k)=(y_k, m^1_k, m^2_k)$ and consider functions
\begin{align*}
h_{init}(l(v))=l(v) \in \mathbb{N}_+, \ \ \ c_{init}(h) = (0, 0, h+1)
\end{align*}
and for subgraphs $S_1$ and $S_2$ with $S_{1,2} = S_1 \cup S_2 \cup (v_a,v_b)$
\begin{align*}
c(S_{1,2}):=r_c(&\{(c(S_1), h(v_a)),(c(S_2), h(v_b))\}, \mathbbm{1}_{S_1 = S_2}) = \\
\big(r(\{(y_1, h(&v_a), m^1_1, m^2_1),(y_2, h(v_b), m^1_2, m^2_2)\}, \mathbbm{1}_{S_1 = S_2}), \ m^2_1+m^2_2+1, \ 2m^2_1 + 2m^2_2 + 2 \big) \\
 %&(1- \mathbbm{1}_{s_1\equiv s_2})e_1+e_2+1]
 %m_1^2 bigger than all h(V(s_1)) 
 %m_2^2 bigger than all h(V(s_2)) 
 %m_{1,2}^1 what to add at application
 %m_{1,2}^2 bigger than all h(V(s_{1,2})) after application
r_v(c(S_{1,2}),h&(v), \mathbbm{1}_{v \in V(S_1)}) =
\left\{\begin{array}{lr}
        h(v)+m^1_{1,2}, \ \ & \text{if } \mathbbm{1}_{v \in V(S_1)}=1 \\
        h(v), \ \ & \text{else} 
        \end{array} \right\}
\end{align*}
% \begin{remark}
% If $r: \{ \mathbb{N}_0^4, \mathbb{N}_0^4 \} \times \mathbb{N}_0 \rightarrow \mathbb{N}_0 $ then all $c(S_k) \in \mathbb{N}_0^3$. Also, in the Appendix we show that $m^1_k=0$ if and only if $|E(S_k)|=0$. Thus, it serves as the required $zero$-symbol.
% \end{remark}
%Consider the following functions:
where
\begin{align*}
&\tau(i,j) = \frac{(i+j)(i+j+1)}{2}+j, \ \ \rho(i,j) = (i+j,ij) \\
&r(y_1,h_1,m_1,n_1,y_2,h_2,m_2,n_2,b) = \tau\big(\tau\big(\rho(\tau^4(y_1,h_1,m_1,n_1), \tau^4(y_2,h_2,m_2,n_2))\big), b\big)
\end{align*}
% \begin{lemma}
% In the above setup, there exists a continuous and bounded function $r:\mathbb{R}^9 \rightarrow \mathbb{R}$ that is injective in $\{ \mathbb{N}^4, \mathbb{N}^4 \}\times \mathbb{N}$. Namely,
% \begin{align*}
% r(&y_1,h_1,m_1,n_1,y_2,h_2,m_2,n_2,b) = \tau\big(\tau\big(\rho(\tau^4(y_1,h_1,m_1,n_1), \tau^4(y_2,h_2,m_2,n_2))\big), b\big)
% \end{align*}
% \label{lemma:injective_r}
% \end{lemma}
In the Appendix we prove that the functions presented in this section satisfy the requirements in Theorem \ref{thm:main}, which allows us to arrive at the following:
\begin{theorem}[NPA Existence Theorem]
There exists functions for Algorithm \ref{alg:NPA} such that every produced graph encoding is iso-injective.
\label{thm:NPAtheorem}
\end{theorem}

\subsection{Corollaries}
\label{sec:corollaries}

% \begin{proof}
% This follows from Theorem \ref{thm:main} and Lemma \ref{lemma:exist}.
% \end{proof}

In our discussion of Algorithm \ref{alg:NPA} we will assume that it uses functions such that Theorem \ref{thm:NPAtheorem} holds. See Appendix for additional corollaries and remarks.

\begin{corollary}
For Algorithm \ref{alg:NPA}, given graphs $G,H \in \mathcal{G}$, $G \simeq H$ if and only if $ \bm{C}([G]) \cap \bm{C}([H]) \neq \emptyset $. I.e. it solves the graph isomorphism problem and canonization.
\end{corollary}

\begin{corollary}
For graphs $G,H \in \mathcal{G}$ consider multiset $ I = W(G) \cap W(H) $. Each $w \in I$ corresponds to a shared subgraph between $G$ and $H$, and $|I|$ is a lower bound to the number of shared subgraphs. The graph corresponding to $I$ is a lower bound (by inclusion) to the largest shared subgraph.
\end{corollary}

\begin{lemma}
\label{lemma:fcountinjective}
Assume $\mathcal{X}$ is countable. There exists a function $f : \mathcal{X} \rightarrow \mathbb{R}^n$ so that $h(X) = \sum_{x \in X} f(x)$ is unique for each multiset $X \subset \mathcal{X}$ of bounded size. Moreover, any multiset function $g$ can be decomposed as $g (X ) = \phi( \sum_{x \in X}f (x))$ for some function $\phi$.
\end{lemma}

% \begin{proof}
% Mapping $Z : \mathcal{X} \rightarrow \mathbb{N}$, and $|X| < N$. Consider $f(x) = N^{-Z(x)}$. See Appendix.
% \end{proof}

\begin{corollary}
If $\mathcal{G}_* \subset \mathcal{G}$ and $\{ |C(G)| \ | \ G \in \mathcal{G}_* \}$ is bounded (number of connected components is bounded), there exists a function $f$ such that any two graphs $G$ and $H$ in $\mathcal{G}_*$ are isomorphic if $\sum_{c \in C(G)} f(c)=\sum_{c \in C(H)} f(c)$.
\label{cor:sum}
\end{corollary}

In the Appendix we show, given a graph isomorphism class $[S]$ and using NPA, a Turing-decidable function for detecting the presence of $[S]$ within a graph $G$; however, if we only have one global encoding for all of $G$ such a Turing-decidable function might not exist. Unless there is some subgraph-information in the encoding we are left to enumerate an infinite set, which is Turing-undecidable. This points to the strength of having the encoding of a graph $G$ coupled with encodings of its subgraphs.

%For bounded graph isomorphism classes $\mathcal{G}_b$ we have: 

% \begin{proofatend}
% For bounded graphs $\mathcal{G}_b$, the multiset $\{w_i \}$ is bounded.
% \end{proofatend}

\subsection{Use of Neural Networks}

\begin{theorem}[NPA Universal Approximation Theorem]
\label{thm:NPAUniversal}
Functions $r_v, r_c, h_{init}, c_{init}$ that satisfies requirements of Theorem \ref{thm:main}, and a function $f_3$ enabling Lemma \ref{lemma:fcountinjective} from Section \ref{sec:corollaries}, can be perfectly approximated by NNs for graphs in $\mathcal{G}_b$ and pointwise approximated for graphs in $\mathcal{G}$.
% Functions $r_v, r_c, h_{init}, c_{init}$ presented in Section \ref{sec:reqfuncs} that satisfies requirements of Theorem \ref{thm:main}, and function $f_3(i)=N^{-i}$ from Section \ref{sec:corollaries}, can be perfectly approximated by NNs for graphs in $\mathcal{G}_b$ and pointwise approximated for graphs in $\mathcal{G}$.
\end{theorem}

%%%%%%%%%%%%%%%%%%%%%%%%%%%%%%%%%%%%%%%%%%%%%%%%%%%%
%%%%%%%%%%%%%%% PUT IN APPENDIX %%%%%%%%%%%%%%%%%%%%
%%%%%%%%%%%%%%%%%%%%%%%%%%%%%%%%%%%%%%%%%%%%%%%%%%%%
% We make use of the following functions:
% \begin{align*}
%     c&_{init}(i) = i+1\\ 
%     f&_1(i,j) = i+j+1 \\
%     f&_2(i,j) = 2i+2j+2 \\
%     r&(y_1,h_1,m_1,n_1,y_2,h_2,m_2,n_2,b) = \\
% & \tau\big(\tau\big(\rho(\tau^4(y_1,h_1,m_1,n_1), \tau^4(y_2,h_2,m_2,n_2))\big), b\big) \nonumber \\
% r&_v(\dots,m,h,1_{ind}) = h+1_{ind}m 
% \end{align*}
% and to a lesser extent
% $$f_3(i) = N^{-i}$$

By Theorem \ref{thm:finiteuniversal}, NNs can perfectly approximate any function on a finite domain so the case of $\mathcal{G}_b$ is straightforward. 
% We know from Theorem [] that NNs can perfectly approximate any function on a finite domain. Since the set of bounded graphs $\mathcal{G}_b$ is finite, we know from previous results that NNs can be used to perfectly approximate any functions on $\mathcal{G}_b$, including $c_{init}$, $f_1$, $f_2$, $r$, $r_v$, and $f_3$.
% \begin{remark}
% For unbounded graphs, we will assume that a bounded injective function like Sigmoid or Tanh can be pointwise approximated by a NN. This is because they are often used as activation functions and help us bound our domains.
% \end{remark}
However, for countably infinite $\mathcal{G}$ the situation is different. Consider functions from Section \ref{sec:reqfuncs} and \ref{sec:corollaries} (Lemma \ref{lemma:fcountinjective}). They are continuous (in $\mathbb{R}^*$) but not bounded, we are applying these functions recursively and would want both the domain and the image to be bounded iteratively. Without losing any required properties we can compose these functions with an injective, bounded, and continuous function with continuous inverse such as Sigmoid, $\sigma$,
%in the following way $f^*=\sigma \circ f \circ \sigma^{-1}$, 
and use $h_{init}(l(v)) = \sigma(l(v))$. Then these functions can be pointwise approximated by NNs. However, recursive application of a NN might increase the approximation error. We use NNs for all non-sort functions. For $r_c$ we use a tree-LSTM \cite{treelstm} and for $r_v$ we use a LSTM. See Appendix for details.

\subsection{A Baseline}

To gauge how conducive our approach is to learning and how important the strict isomorphic properties are, we present a simpler and non iso-injective baseline model which is the same as Algorithm \ref{alg:NPA} but the second outer for-loop has been replaced by Algorithm \ref{alg:NPBA}. Some results of this algorithm can be seen in Table \ref{table:results} and it performs at state-of-the-art.

\begin{algorithm}[h]
\caption{Node Parsing Baseline Algorithm (NPBA)}
\label{alg:NPBA}
%\SetAlgoLined
\begin{algorithmic}
% \KwResult{ A sequence of feature vectors $w_1,...,w_m$. }
 \FOR{i=1, \dots, m}
 \STATE Let $(v_a,v_b) = e_i$ and let $S_1, S_2 \in A_i$ be subgraphs with $v_a \in S_1$ and $v_b \in S_2$\; 
 %\STATE Let $S_1, S_2 \in A_i$ be subgraphs with $v_a \in S_1$ and $v_b \in S_2$\;
 %\STATE Let $S_2 \in A_i$ be subgraph with $v_b \in S_2$\;
 \STATE $c(S_{1,2}) = r_c(\{c(S_1), c(S_2) \}) $\;
 \ENDFOR
 \end{algorithmic}
\end{algorithm}

%\begin{theorem}
%(OPEN PROBLEM) If for all orderings of edges for two graphs $G$ and $H$, Algorithm \ref{alg:NPBA} gives the same result, does it imply that $G \simeq H$?
%\end{theorem}

\subsection{Class-Redundancy, Sorting, Parallelize, and Subgraph Dropout}

% In our algorithmic approach we do not construct a function on isomorphism classes $\bm{\mathcal{G}}$ because two graphs $G,H \in \mathcal{G}$ with $G\simeq H$ might have different encondings: $C(G) \neq C(H)$. Rather we are dealing with a multivalued function on $\bm{\mathcal{G}}$ and the size of set $\bm{Alg}(\bm{G})$ is called the class-redundancy of $\bm{G}$.

The class-redundancy in the algorithm and functions we propose enters at the sort functions $s_e$ (sorts edges) and $s_v$ (sorts nodes within edges). 
%Since identical ordering of the edges and the nodes within edges give the same encoding, 
Thus, a loose upper bound on the class-redundancy is $O((m!)2^m)$.
%Alternatively, we can simply use an ordering on the nodes $0(n!)$ if it is smaller. 
%However, a more exact upper bound is $O((t_1!)(t_2!)\dots(t_k!)(2^p))$, where $t_i$ are the sizes of the consecutive ties for the sorted edges, and $p$ (bounded by $m$) is the number of ties for the sorting of nodes within edges. 
A better upper bound is $O((t_{1,1}!)\dots(t_{1,l_1}!)(t_{2,1}!)\dots(t_{k,l_k}!)(2^p))$
where each $t_{i,j}$ is the number of ties within group $j$ of groups of subgraphs that could be connected within the tie $i$. The order in between disconnected tied subgraph groups does not affect the output. 
%Again, $p$ is the number of edges where $s_v$ has ties .
See Appendix for \#edge-orders, i.e. $O((t_{1,1}!)\dots(t_{1,l_1}!)(t_{2,1}!)\dots(t_{k,l_k}!))$, on some datasets.  

% We could also make the sorting of the nodes in the edges be a function of more variables, e.g. the number of edges in each of the two components, or an injective function of the $c$-values, so that there will only be a tie for two components that are isomorphic. The analysis however gets tricker.

We focus on function $s_e$. Each edge can be represented by the following vector [\textit{deg1}, \textit{deg2}, \textit{label1}, \textit{label2}]. We assume \textit{deg1}, \textit{deg2} as well as \textit{label1}, \textit{label2} are in descending order, and that ties are broken randomly. This work makes use of four $s_e$ functions: (i) \textit{none}: Does not sort at all. (ii) \textit{one-deg}: Sorts by \textit{deg1}. (iii) \textit{two-degs}: Sorts lexicographically by \textit{deg1}, \textit{deg2}. (iv) \textit{degs-and-labels}: Sorts lexicographically by \textit{deg1}, \textit{deg2}, \textit{label1}, \textit{label2}.

%\subsubsection{Parallelize and Levels}

Since the encodings of subgraphs that share no subgraph  
%(under a certain edge-ordering)
do not affect each other, we can parellalize our algorithm to encode such subgraphs in parallel. For example, a graph of just ten disconnected edges can be parellalized to run in one step. We call the number of such parellalizable steps for a graph's \textit{levels}. See Appendix for \#levels on some datasets.

%\subsubsection{Subgraph Dropout}

%The iso-injective encodings produced for a graph $G$ may depend on the ordering of the edges. % and nodes. 
%This is suboptimal but a NN could learn to map the finite number of different encodings  $\{ Alg(H) \ | \ H \in \mathcal{G}, \bm{H} = \bm{G} \}$ to the same encoding. 
%At the same time,
%, and thus produces encodings for different subgraphs that may depend on what exact order of edges that was used.
% Different orderings of edges causes the algorithm to visit and encode different subgraphs. f the algorithm is run on all possible orderings on the edges, it will produce an encodings for each possible subgraph.

In most cases, one run of NPA on graph $G$ computes features for a very small portion of all subgraphs of $G$. We could run NPA on all possible orders to make sure it sees all subgraphs, but this is very costly. Instead, we use the random sample of featurized subgraphs as a type of dropout \cite{dropout}. During training, at each run of the algorithm we use only one ordering of the edges, which discourages co-adaptation between features for different subgraphs. At testing, we let the algorithm run on a sample of $K$ orderings, and then average over all these runs. We call this technique \textit{subgraph dropout}. 

% \subsection{Benefits of the features}

% The algorithm construct a feature for each processed subgraph, and each new features is a combination of two previous features. This local-to-global procedure arguably is conducive to learning. 

% \subsubsection{Interpretability}

% Since each feature corresponds to an isomorphism class of a subgraph, we can get insights into what subgraphs make a graph belong to a certain class in a dataset. 

%See which subgraps/motifs that algorithm picks up on.

% \subsection{Neural Networks}

% We use NNs for all non-sort functions. For NPA and NPBA we use for $r_c$ a tree-LSTM similar to those in \cite{treelstm} and for $r_v$ we use a LSTM. See Appendix for details.

\section{Experiments}
\label{sec:experiments}

See Table \ref{table:results} for results on graph classification benchmarks.
We report average and standard deviation of validation accuracies across the 10 folds within the cross-validation.
%We use the Adam optimizer with initial learning rate 0.01 and decay the learning rate by 0.5 every 50 epochs.
%We tune the number of epochs as a hyper-parameter, i.e., a single epoch with the best cross-validation accuracy averaged over the 10 folds was selected.
%Note that due to the small dataset sizes, an alternative setting, where hyper-parameter selection is done using a validation set, is extremely unstable, e.g., for MUTAG, the validation set only contains 18 data points.
In the experiments, the $W(G)$ features are summed and passed to a classifier consisting of fully connected NNs. %classify-NN consisting of fully connected NNs.
% either one fully-connected layer and a readout layer (for MUTAG, PTC, and PROTEINS) or two fully-connected layers and a readout layer (for NCI1), where the hidden-dim of the fully connected layers is of size $d_{hidden}$. For $h_{init}$ we use a linear-layer followed by a batchnorm (for MUTAG, PTC, and PROTINES) or a linear-layer followed by activation function and batchnorm (for NCI1). In addition, for NCI1 we used dropout=0.2 after each layer in the classify-net and on the elemtns of $W(G)$ before summing them.
% Also, in our experiments we skipped including the $w_i$ features for the single nodes. In fact, all datasets consist of connected graphs.
% For the NPBA tree-lstm the dimensions of $c^1$ and $c^2$ is $d_{hidden}$. For the NPA the dimensions of $c^1$ and $c^2$ is $d_{hidden}$ and the dimension of $h$ is $d_{hidden}/2$. 
% We used the following settings for $d_{hidden}$ and batch size:
% \begin{itemize}
%     \item PTC, PROTEINS, and MUTAG we used $d_{hidden}=16$, and batch-size=$32$.
%     \item NCI1 we used $d_{hidden}=64$, and batch-size=$128$.
% \end{itemize}
For NPA, $s_v$ sorts randomly, but with "-S", $s_v$ sorts based on the levels of subgraphs $S_1$ and $S_2$. For subgraph dropout "-D" we use $K=5$. 
The four bottom rows of Table \ref{table:results} compare different functions for sorting edges ($s_e$).
%For details, see Appendix.

\begin{table*}%[h]
\caption{GNN is best performing variant from \cite{powerful}. 
%PatchySan \cite{Patchysan}. DCNN \cite{DCNN}. DGCNN \cite{DGCNN}. 
%The "-D" is short for "with subgraph dropout." The "-S" is short for "with $s_v$ level-sort". 
*: Best result with and without subgraph dropout.}
\begin{center}
\begin{tabular}{l l l l l}
\hline
Datasets: & NCI1 & MUTAG & PROTEINS & PTC  \\
\# graphs: & 4110 & 188 & 1113 & 344 \\
\# classes: & 2 & 2 & 2 & 2 \\
%Avg \# nodes: & 30 & 18 & 39 & 26 \\
%Avg \# edges: & 32 & 20 & 74 & 26 \\
%$O($avg \# edge-orders$)$: & 10^{46} & 10^8 & 10^{143} & 10^{69} \\
%$O($avg \# class-redundancy$)$: & 10^{54} & 10^{13} & 10^{211} & 10^{79} \\
% $O($avg \# edge-orders$)$: & 10^{30} & 10^9 & 10^{143} & 10^{25} \\
% $O($avg \# class-redundancy$)$: & 10^{39} & 10^{14} & 10^{211} & 10^{37} \\
%$O($median \# edge-orders$)$: & 4 & 2 & 512 & 2 \\
%$O($median \# class-redundancy$)$: & 4096 & 256 & 10^{6} & 128 \\

\hline
%\hline
%Theirs & 0 & 0 & 0 & 0  \\
PatchySan \cite{Patchysan} & 78.6$\pm$1.9 & 92.6$\pm$4.2 & 75.9$\pm$2.8 & 60.0$\pm$4.8 \\
DCNN \cite{DCNN} & 62.6 & 67.0 & 61.3 & 56.6 \\ 
DGCNN \cite{DGCNN} & 74.4$\pm$4.7 & 85.8$\pm$1.6 & 75.5$\pm$0.9 & 58.6$\pm$2.5 \\ 
GNN \cite{powerful} & 82.7$\pm$1.7 & 90.0$\pm$8.8 & 76.2$\pm$2.8 & 66.6$\pm$6.9 \\
NPBA (ours) & 81.0$\pm$1.1 & 92.8$\pm$6.6 & 76.6$\pm$5.7 & 67.1$\pm$5.9 \\
NPBA-D (ours) & 83.7$\pm$1.5 & 92.2$\pm$7.9 & \textbf{77.1$\pm$5.3} & 65.5$\pm$6.8 \\
NPA (ours) & 81.8$\pm$1.9 &  92.8$\pm$7.0 & 76.9$\pm$3.0 & \textbf{67.6$\pm$5.9} \\
NPA-D (ours) & \textbf{84.0$\pm$2.2} & 92.8$\pm$7.5 & 76.8$\pm$4.1 & 67.1$\pm$6.9 \\
NPA-S (ours) & 81.5$\pm$1.6 &  \textbf{93.3$\pm$6.0} & 76.5$\pm$5.0 & 65.9$\pm$8.3 \\
NPA-D-S (ours) & 83.0$\pm$1.2 & \textbf{93.3$\pm$6.0} & 76.3$\pm$4.5 & 66.2$\pm$7.7 \\
\hline
NPA* (degs-and-labels) & 83.2$\pm$1.6 & 88.9$\pm$10.5 & 75.9$\pm$5.4 & 63.2$\pm$6.3  \\
NPA* (two-degs)  & \textbf{84.0$\pm$2.2} & 91.7$\pm$6.7 & 76.2$\pm$4.6 & \textbf{67.6$\pm$5.9}  \\
NPA* (one-deg)  & 79.2$\pm$1.9 & 92.8$\pm$7.0 & 76.5$\pm$4.9 & 64.7$\pm$7.0  \\
NPA* (none)  & 77.7$\pm$3.0 & 92.8$\pm$7.5 & 76.9$\pm$3.0 &  65.3$\pm$5.9 \\
\hline
\end{tabular}
\end{center}
\label{table:results}
\end{table*}

\begin{table*}%[h]
\caption{(Train-accuarcy). Comparing NPA against other methods for certain types of graphs.
%GNN (GIN) from \cite{powerful}.
}
\begin{center}
\begin{tabular}{l l l l l l}
\hline
Datasets: & GNN-Hard & NPBA-Hard & Erdos & Erdos-Labels & Random-Regular  \\
\# graphs: & 32 & 36 & 30 & 100 & 10 \\
\# classes: & 2 & 2 & 30 & 100 & 10 \\
Avg \# nodes: & 17$\pm$9 & 1.5$\pm$0.5 & 10$\pm$0 & 10$\pm$0 & 8$\pm$0 \\
Avg \# edges: & 34$\pm$19 & 21$\pm$10 & 45$\pm$7 & 45$\pm$7 & 16$\pm$0 \\
$O$(median \\ \# edge-orders): & $10^{35}$ & $10^{21}$ & $10^{19}$ & $10^{8}$ & $10^{10}$ \\
\hline
%\hline
%Theirs & 0 & 0 & 0 & 0  \\
GNN (GIN) \cite{powerful} & 50 & \textbf{100} & \textbf{100} & \textbf{100} & 10 \\
NPBA (ours) & \textbf{100} & 50 & 83 & \textbf{100} & 70 \\
NPA (ours) & \textbf{100} & \textbf{100} & \textbf{100} & \textbf{100} & \textbf{90} \\
\hline
\end{tabular}
\end{center}
\label{table:synthetic}
\end{table*}

\subsection{Synthetic Graphs}

We showcase synthetic datasets where the most powerful GNNs are unable to classify the graphs, but NPA is. See Appendix for related discussion and Table \ref{table:synthetic} where
\begin{enumerate}
    \itemsep0em 
    \item GNN-Hard: Class 1: Two disconnected cycle-graphs of $n/2$ vertices. Class 2: One single cycle-graph of $n$ vertices. ($n=2,4,6,\dots,32$)
    \item NPBA-Hard: Class 1: Two nodes with $m$ edges in between. Class 2: Two nodes, with $m$ self-edges from one of the nodes. ($m=2,3,4,\dots,19$)
    \item Erdos: Random Erdos-Renyi graphs.
    \item Random-Regular: Each node has the same degree with configuration model from \cite{randomreg}.
\end{enumerate}

% \begin{enumerate}
%     \item GNN-Hard: Deciding weather a graph of $n$ vertices consists of two disconnected cycle-graphs of $n/2$ vertices or of one single cycle-graph of $n$ vertices.
%     \item NPBA-Hard: Class 1: Graphs consisting of 2 nodes with $n$ edges in between. Class 2: Graphs of two nodes, with $m$ self-edges from one of the nodes.
%     \item Erdos-Renyi Graphs: Random Erdos-Renyi graphs.
%     \item Random Regular Graphs: Each node has the same degree with configuration model from [].
% \end{enumerate}

% GNNs work very well on Erdos Random Graphs, while the NPA could get 100\% training accuracy, it took many more epochs to do so compared to the GNNs.

\section{Discussion}

In this paper, we develop theory and a practical algorithm for universal function approximation on graphs. Our framework is, to our knowledge, theoretically closest to a universal function approximator on graphs that performs at the state-of-the-art on real world datasets. It is also markedly different from other established methods and presents new perspectives such as subgraph dropout. In practice, our framework shares weaknesses with GNNs on regular graphs, and we do not scale as well as some other methods. Future work may reduce the class-redundancy, explore bounds on expected class-redundancy, modify GNNs to imbue them with iso-injective properties, or combine iso-injective encodings (from NPA) with invariant encodings (from GNNs) to enable the best of both worlds.
%approximately injective
% At the same time, there are several weaknesses to our framework. We took great care to consider the set of unbounded finite graph isomorphism classes, but our method does not scale as well as several other graph learning frameworks. We compare our method to GNNs which are powerful graph learners, but in the current form, our framework share some of its weakness when it comes to regular graphs. 
% %That is, when all nodes have the same degree, the sorting of edges will be less strict, leading to a greater class-redundancy. 

% There are aspect of the work that deserves greater inspection. Our assumptions on the functions that NNs can approximate are fairly general and understanding the gap this and the functions that the actual NNs we use in our experiments can learn needs to be more carefully understood.

% Future work could focus on decreasing the class-redundancy and perhaps explore bounds on expected class-redundancy. One could also attempt to modify GNNs to imbue them with iso-injective properties. In addition, combining iso-injective encodings (such as from NPA) with approximately injective encodings (such as from GNNs) might enable, as it were, the best from both worlds. 

\section{Broader Impact}

This work helps advance the fields of machine learning and AI, which as a whole is likely to have both positive and negative societal consequences \cite{broad1, Brundage2016}; many of which might be unintended \cite{unintended}. The coupling of application and theory in this work aims at improving human understanding of AI which is related to efforts within for example explainable AI \cite{explainables}. Such efforts may reduce unintended consequences of AI. 

%This work does not present any foreseeable societal consequence.

%This work has the following potential positive impact in the society…. At the same time, this work may have some negative consequences because… Furthermore, we should be cautious of the result of failure of the system which could cause..

\section{Acknowledgements}

This work was supported by Altor Equity Partners AB through Unbox AI
(www.unboxai.org). I am grateful for Bradley J. Nelson's help in reading the paper and for his suggestions on how to make it clearer. I also want to express my greatest gratitude to Professor Gunnar Carlsson and Professor Leonidas Guibas for their unwavering support and belief in me.

%%%%%%%%%%%%%%%%%%%%%%%%%%%%%%%%%%%%%%%%%%%%%%%%%%%%%%%%%%%%%%%%%%%%%%%%%%%%%%%
%%%%%%%%%%%%%%%%%%%%%%%%%%%%%%%%%%%%%%%%%%%%%%%%%%%%%%%%%%%%%%%%%%%%%%%%%%%%%%%

\bibliography{ref}
% \bibliography{ICML}
% \bibliographystyle{icml2020}

%\end{document}

\clearpage
\section*{Appendices}
\appendix

\section{Theory}

\subsection{Preliminaries: Additional Definitions, Remarks, and Proofs}

\subsubsection{Additional Definitions and Remarks}
We add the following definitions:

\begin{definition}
A \textit{subgraph} $S$ of a graph $G$, denoted $S \subset G$, is another graph formed from a subset of the vertices and edges of $G$. The vertex subset must include all endpoints of the edge subset, but may also include additional vertices. 
\end{definition}

\begin{definition}
We denote the disjoint union between two sets $A,B$ as $A \sqcup B$.  %\bjn{in topology, it is standard to use $\sqcup$}
\end{definition}

\begin{definition}
We denote the set-builder notation for multisets as $[x \ | \ Predicate(x) ]$, i.e. with brackets to emphasize it constructs a multi-set.   %\bjn{in topology, it is standard to use $\sqcup$}
\end{definition}

\begin{definition}
If we write $f(A)$ where $A$ is a subset of the domain of $f$, we mean the multiset $f(A) := [f(x) \ | \ x\in A ]$. % (the subscript $M$ denotes multiset construction).
\end{definition}

\begin{definition}
Let $f:X\to Y$ be a function from a set $X$ to a set $Y$. If a set $A$ is a subset of $X$, then the restriction of $f$ to $A$ is the function
$$f|_{A}:A\to Y$$
given by $f|_{A}(x) = f(x)$ for $x$ in $A$. Informally, the restriction of $f$ to $A$ is the same function as $f$, but is only defined on $A \cap dom(f)$.
\end{definition}

\begin{definition}
For an iso-injective function $f:\mathcal{G} \rightarrow Y$ we define the \textit{iso-inverse} as the function $\bm{f}^{-1}:\operatorname{im}(f) \rightarrow \bm{\mathcal{G}}$, where $\operatorname{im}(f) = \{ y \ | \ y\in Y, \exists G \in \mathcal{G}, f(G)=y \}$, as
$$ \bm{f}^{-1}(y) = [G], \exists G \in \mathcal{G}, f(G)=y  $$
%which is well-defined by the iso-injective property of $f$.
\end{definition}

\begin{definition}
The \textit{subgraph isomorphism problem} consists in, given two graphs $G$ and $H$, determining whether $G$ contains a subgraph that is isomorphic to $H$.
\end{definition}

\begin{definition}
With a function $f:X \rightarrow Y$ being injective across domains $X_1$ and $X_2$ with $X_1, X_2 \subset X$, we mean that for all $x_1 \in X_1, x_2 \in X_2$ with $f(x_1)=f(x_2)$ we have $x_1=x_2$.
\end{definition}

\begin{definition}
In some proofs we say \textit{subgraph $S$ encoded at step $j$} of Algorithm \ref{alg:NPA} (NPA), with which we mean that if $j=0$ then $S$ is a single node that is encoded in the first for loop of NPA, and if $j>0$ then $S$ contains an edge and is encoded in the second for loop of NPA with $j=i$.
\end{definition}

We also add the following remarks:

%%% REDUNDANT?
\begin{remark}
Functions on nodes $f:V(G) \rightarrow Y$, such as node labels, are functions of graphs too, because it makes no sense to compare indices or nodes between different graphs that are not subgraphs of the same graph. That is, each such function is different for each graph $G$, so if we abuse notation when having also a graph $H$ and $f:V(H) \rightarrow Y$ in a shared context with $G$, then $v_1=v_2$ implies $f(v_1)=f(v_2)$ only if $v_1,v_2 \in V(G)$ or $v_1,v_2 \in V(H)$. Similarly, intersection between edges or nodes of two graphs $S_1$ and $S_2$ is only interesting to us if $S_1, S_2$ are subgraphs of some graph $G$.
\end{remark}

\begin{remark}
We can bound any iso-injective function $Alg: \mathcal{G} \rightarrow \mathbb{R}^d$ by composing (this simply forces the convergent subsequence to be the values in $\mathbb{R}^d$ with increasing norm) with the injective and continuous Sigmoid function $ \sigma(x) = \frac{1}{1+e^x} $.
\end{remark}
\subsubsection{Proof of Lemma \ref{lemma:graphsetscardinality}}
\begin{proof}%[Proof of Lemma \ref{lemma:graphsetscardinality}]
For each $n \in \mathbb{N}_+$ there is a finite number of graphs $G$ with $|V(G)|+|E(G)|+\sup_{v\in V(G)}(l(v)) = n$, and a countable union of countable sets is countable. Similarly, bounded graphs means that such a $n$ is bounded by $b$, and a finite union of finite sets is finite. Furthermore, $|\bm{\mathcal{G}}|\leq |\mathcal{G}|$ and $|\bm{\mathcal{G}}_b|\leq |\mathcal{G}_b|$.
\end{proof}

\subsubsection{Proof of Theorem \ref{thm:fromisotoinjective}}
\begin{proof}
Consider, $h = (f \circ \bm{g}^{-1}): \bm{\mathcal{G}}\rightarrow Y$ which is well defined since $\bm{g}^{-1}$ is a function on $\operatorname{im}(g)$, and $f=h \circ g$.
\end{proof}

\subsubsection{Proof of Theorem \ref{thm:recurrentuniversalapproximationtheorem}}
\begin{proof}
See \cite{SIEGELMANN1995132} for proof.
\end{proof}

\subsection{Bounded Graphs}

\subsubsection{Proof of Theorem \ref{thm:finiteuniversal}}
\begin{proof}
In \cite{arora2018understanding} it is proven that any continuous piecewise linear function is representable by a ReLU NN, and any finite function can be perfectly approximated by a continuous piecewise linear function.
\end{proof}

\subsubsection{Proof of Theorem \ref{thm:finiteboundediso}}
\begin{proof}
Consider the function $g:\operatorname{im}(Alg) \rightarrow \mathbb{R}^d$:
$$ g(x) = (f \circ \bm{Alg}^{-1})(x)  $$
Which is well-defined because both $f$ and $\bm{Alg}^{-1}$ are functions on their respective domains. Since $\operatorname{im}(Alg)$ is a finite subset of $\mathbb{R}^d$ we know there is a NN $\varphi$ that perfectly approximates $g$, and thus we have 
$$f = \varphi \circ Alg$$
\end{proof}

\subsection{Unbounded Graphs}

%\subsubsection{Proof of Theorem \ref{thm:noboundedplinear}}

\subsection{On Remark \ref{remark:converge}}

Suppose $Alg:\mathcal{G} \rightarrow \mathbb{R}^d$ is an iso-injective function and $\varphi: \mathbb{R}^d \to \mathbb{R}$ is a NN. We analyze the functions $f:\bm{\mathcal{G}} \rightarrow \mathbb{R}$ that $\varphi \circ Alg$ can approximate. By Theorem \ref{thm:universal}, if $\operatorname{im}(Alg) \subset \mathbb{R}^d$ is bounded, then $\varphi$ can approximate all continuous functions on the closure $\overline{\operatorname{im}(Alg)}$. Since $\mathcal{G}$ is countably infinite, we may consider the sequence $\operatorname{im}(Alg) =(\bm{Alg}([G]_i)_{j=0}^{k_i})_{i=0}^{\infty} = ((a_i)_{j=0}^{k_i})_{i=0}^{\infty} \subset \mathbb{R}^d$. From the Bolzano-Weierstrass Theorem we know \textit{every bounded sequence of real numbers has a convergent subsequence}. 
If $\operatorname{im}(Alg)$ is bounded then so is $((a_i)_{j=0}^{k_i})_{i=0}^{\infty}$, and thus it has a convergent subsequence. Similarly, the subsequence $Alg([G]_{i=0}^{\infty})$ with $Alg([G]_i) = Alg(H), H \in [G]_i$, corresponding to a sequence over the graph isomorphism classes $[G]_i \in \bm{\mathcal{G}}$, has a convergent subsequence. Meaning that for every $\delta > 0$ there is a countably infinite set $A \subset \bm{\mathcal{G}}$ such that $[G]_i,[G]_j \in A$ implies $||Alg([G]_i)-Alg([G]_j)|| < \delta$. Let $L$ denote the limit point of one such convergent subsequence. By Theorem \ref{thm:universal}, we assume that $\varphi$ can approximate only continuous functions, this means for every $\epsilon > 0$ there exists a $\delta > 0$ such that that $||L-Alg([G])|| < \delta$ with $[G] \in \bm{\mathcal{G}}$ implies $||\varphi(L) - \varphi(Alg([G]))|| < \epsilon$. 
% I.e. the properties of the iso-injective function $Alg$ deeply affects the set of functions on graph isomorphism classes that can be approximated by $\varphi \circ Alg$. 
Note that the same holds for an injective function $h:\bm{\mathcal{G}} \rightarrow \mathbb{R}^d$, because the sequences $\operatorname{im}(h) = h([G]_{i=0}^{\infty})$ and $((a_i)_{j=0}^{k_i})_{i=0}^{\infty}$ have the same cardinality.

\subsection{Theorems and Proofs}

\begin{theorem}
\label{thm:noboundedplinear}
There is no finite width and depth NN with bounded or piecewise-linear activation function that can pointwise approximate an unbounded continuous function on an open bounded domain.
\end{theorem}
\begin{proof}
Such NNs must be bounded on bounded domains.
\end{proof}

\begin{theorem}
\label{thm:nofinitenoneinifintediff}
There is no finite width and depth NN with an activation function $\sigma$ and $k \geq 0$ such that $\frac{d^k \sigma}{d x^k}=0$ that can pointwise approximate all continuous functions on unbounded domains.
\end{theorem}
%\subsubsection{Proof of Theorem \ref{thm:nofinitenoneinifintediff}}
\begin{proof}
Consider $f(x) = x^{k+1}$ such that $\frac{d^k f}{x^k} \neq 0$. The NN cannot asymptotically approximate $f$.
\end{proof}

\begin{theorem}[Bolzano-Weierstrass]
Every bounded sequence of real numbers has a convergent subsequence.
\label{thm:bw}
\end{theorem}
\begin{proof}
Well-known result, see \href{https://en.wikipedia.org/wiki/Bolzano%E2%80%93Weierstrass_theorem}{Wikipedia} or your favorite analysis book.
\end{proof}

\subsubsection{Proof of Theorem \ref{thm:universal}}
\begin{proof}
Proof can be found in \cite{SONODA2017233} and \cite{universalthesis} for a large family of activation functions.
\end{proof}

%\subsubsection{Proof of Theorem \ref{thm:bw}}

\subsubsection{Proof of Theorem \ref{thm:2ndapproxthm}}
\begin{proof}
If $X$ is closed, it follows immediately from Theorem \ref{thm:universal}. Suppose $X$ is open, then we know by Theorem \ref{thm:universal} that $\varphi$ can pointwise approximate $f$ on a compact set, but since $f$ is bounded we know that each limit point is finite. Thus, we may just add them and define $g$ as $f$ extended with the limit points. Then $g$ is continuous on a compact $\overline{X}$, so $\varphi$ pointwise approximates $g$, but this means it also pointwise approximates $f$.
\end{proof}

\subsection{Algorithmic Idea}

\subsubsection{Proof of Theorem \ref{thm:outline2}}
\begin{proof} 
Suppose Algorithm \ref{alg:outline2} is run on graphs $G$ and $G^*$. Suppose also that the assumptions of the theorem holds for both runs and that $c(S_{1,2})=c(S^{*}_{1,2})$ with $S_{1,2} \subset G, S^*_{1,2} \subset G^*$. This means, since $p>1$ that we can split up in the following way, $S_{1,2} = S_1 \cup S_2$ with $S_1,S_2 \in A \subset G$ and $S^*_{1,2} = S^*_1 \cup S^*_2$ with $S^*_1, S^*_2 \in A^* \subset G^*$. We want to show that $S_{1,2} \simeq S^*_{1,2}$. 

We know since $r$ is injective that 
\begin{align}
 c(S_1)=c(S^{*}_1)&, \ c(S_2)=c(S^{*}_2), \\
 \{ l(v) \ | \ v \in V(S_1) \cap V(S_2) \} &= \{ l(v) \ | \ v \in V(S^*_1) \cap V(S^*_2) \} 
 \label{eq:p11}
\end{align}
(If instead $c(S_1) = c(S^{*}_2), c(S_2)=c(S^{*}_1)$ we can just relabel) This means that there exists isomorphisms $\phi_1:S_1 \rightarrow S^*_1$ and $\phi_2:S_2 \rightarrow S^*_2$.

Consider the following map:
\begin{equation}
    \phi(v) =
    \begin{cases*}
      \phi_1(v) & if $v \in V(S_1)$ \\
      \phi_2(v)  & otherwise
    \end{cases*}
\end{equation}

We set $I = V(S_1) \cap V(S_2)$. Now, since both $\phi_1$ and $\phi_2$ are isomorphisms we know that $\phi$ respects $l$-values, and the only part of the domain where $\phi$ might not respect edges is in $I$. Now let $I^* = V(S^*_1) \cap V(S^*_2)$. 

All values in $l(I)$ are unique among $l(V(S_1)\cup V(S_2))$, all values in $l(I^*)$ are unique among $l(V(S^*_1)\cup V(S^*_2))$. From Equation \ref{eq:p11} we know that $l(I)=l(I^*)$. Suppose $v \in I$ then $\phi_1(v)=\phi_2(v)$ because else $l(\phi_1(v))\neq l(\phi_2(v))\rightarrow l(v) \neq l(v)$ by the stated uniqueness of the $l$-values of $I$ and $I^*$. Since, $\phi_1$ and $\phi_2$ agree on the intersection $I$ we know that all edges must be respected by $\phi$ by construction.

Now we want to show that $\phi$ is a bijection. From construction we know that $\phi$ is a bijection on $V(S_1) \rightarrow V(S^*_1)$. Now $V(S_1) \cup V(S_2) = V(S_1) \sqcup (V(S_2)-I) $ and $V(S^*_1) \cup V(S^*_2) = V(S^*_1) \sqcup (V(S^*_2)-I^*)$. From before we know that $\phi(I) = I^*$. Thus, we know that $\phi$ is injective map on $V(S_2)-I \rightarrow B \subset V(S^*_2)-I^*$ because $\phi$ is equivalent to $\phi_2$ on that domain. To see this, suppose $v \in V(S_2)-I$ and $\phi(v) \in V(S^*_1)$, then we must have $\phi(v) \in I^*$ (since $\phi(v)=\phi_2(v) \in V(S^*_2)$), but this would mean that $v \in I$ (else $l$-value cannot be respected by uniqueness) and we would get a contradiction. Lastly, since $|V(S_2)-I| = |V(S_2)|-|I|$, $|V(S_2)| = |V(S^*_2)|, |I|=|I^*|$, and $|V(S^*_2)-I^*|=|V(S^*_2)|-|I^*|$ we have 
$$|V(S_2)-I| = |V(S^*_2)-I^*|$$
and $\phi$ must be bijective on $V(S_2)-I \rightarrow V(S^*_2)-I$. Thus, $\phi$ is a bijection on $V(S_1)\cup V(S_2) \rightarrow V(S^*_1)\cup V(S^*_2)$.

We are done.

\end{proof}

\section{Method}

\subsection{Algorithm}

%\subsubsection{Proof of Lemma \ref{lemma:disjoint}}
\begin{proof}[Proof of Lemma \ref{lemma:disjoint}]
Since the algorithm processes subgraphs by adding one edge at a time, the theorem follows from proving that at any step in the algorithm, each subgraph in $A_i$ is disjoint and connected, then an edge can only be between two disjoint connected subgraphs or within the same connected subgraph. We prove this by induction on the number of processed edges.

\textit{Base case}: $i=1$. Clearly, all subgraphs consisting of a single vertex are disjoint and each such subgraph is trivially connected.

\textit{Inductive case}: Assume true for $i \geq 1$, we want to show it is true for $i+1$. Now at step $i+1$, by our inductive hypothesis, all subgraphs in $A_i$ are disjoint. The next set of subgraphs $A_{i+1} = (A_i - \{S_1,S_2\}) \cup S_{1,2}$ where $S_{1,2} = S_1 \cup S_2 \cup (v_a,v_b)$, $v_a \in V(S_1)$, and $v_b \in V(S_2)$, is constructed by processing an edge $(v_a,v_b)$. Regardless of whether this edge connects two disjoint subgraphs or is within the same subgraph, in the next step, all subgraphs in $A_{i+1}$ will still be disjoint. This is because we add the new subgraph $S_{1,2}$ to form $A_{i+1}$ but remove the single subgraph (if $S_1 = S_2$) or the two subgraphs (if $S_1 \neg S_2$), to form $A_{i+1}$, that $S_{1,2}$ was connected to by the processed edge. I.e. we remove all subgraphs from $A_i$ (to form $A_{i+1}$) that the new subgraph in $A_{i+1}$ connects to. Also, since each graph $S_1$ and $S_2$ is connected, so must $S_{1,2}$ be by virtue of edge $(v_a,v_b)$.

The lemma follows.
\end{proof}

\begin{remark}
NPA produces a sequence of encodings for a graph $G$ but when finished, set $A_{m+1}$ contains each of the largest (by inclusion) disjoint connected subgraphs of $G$. Since NPA builds encodings recursively from disjoint subgraphs, NPA constructs encodings for each such largest subgraph independently as if it is run once for each of them. Thus, proving that NPA produces iso-injective encodings for connected graphs, implies each multiset $W(G)$ and $C(G)$ is iso-injective also for disconnected graphs. 
\end{remark}

\begin{lemma}
\label{lemma:hvaluesnotchange}
For any graph $S$ encoded at step $i$ on run $G$ on NPA, the function $h^j$ restricted to $V(S)$ does not change from $j=i+1$ up to and including step $k$ (i.e. $j=k$) where $S$ is still a member of $A_{k}$. % and used to encode graph $S_{1,2} = S \cup S^* \cup (v_a,v_b)$.  
\end{lemma}
\begin{proof}
From the description of NPA we can tell that when a graph $S$ is encoded at step $i$ on run $G$, all $h^i$-values of $V(S)$ are updated to $h^{i+1}$-values, while all $h^{i+1}$-values of $V(G)-V(S)$ are inherited from $h^{i}$, and $S$ is added to $A_{i+1}$. Since all graphs in $A_k$ are disjoint (Lemma \ref{lemma:disjoint}), the next time $h$-values of $V(S)$ will change is at step $k'$ when NPA picks $S$ from $A_{k'}$ to encode some subgraph $S_{k'} = S \cup S_2 \cup (v_a,v_b)$, updates $h^{k'+1}$-values of $V(S_{k'})$ with $V(S) \subset V(S_{k'})$, and does not include $S$ in set $A_{k'+1}$ (and never will again). On the other hand, if $S$ is not picked from $A_{k'}$ to encode $S_{k'}$ we know that $V(S_{k'}) \cap V(S) = \emptyset$ by Lemma \ref{lemma:disjoint} so that $h$-values of $V(S)$ do not change, i.e. $h^{k'+1}|_{V(S)}=h^{k'}|_{V(S)}$, and that $S \in A_{k'+1}$.
\end{proof}

\subsubsection{Proof of Theorem \ref{thm:main}}
\begin{proof}[Proof of Theorem \ref{thm:main}]
So we want to show that any two graphs $S_{1,2}$ run $G$ and $S^*_{1,2}$ run $G^*$ with $c(S_{1,2})=c(S^{*}_{1,2})$ are ismorphic.
We prove this by double induction on the number of steps of the algorithm. This is because we need to be able to compare $c$-values that are produced at different runs of the algorithm. I.e. we want to prove a property $P(i,j)$ for all $i,j \in \mathbb{N}$, where $i$ and $j$ reflects step $i$ on first run ($G$) and step $j$ on second run ($G^*$) respectively. By the symmetry of the property, we only need to prove $P(1,1)$ and $P(i,j)\rightarrow P(i+1,j)$. 

To be exact, the property $P(i,j)$ that we will prove consists of the following: that for any subgraph $S$ encoded at step $i' \leq i$ on run $G$ and any sugraph $S^*$ encoded at step $j' \leq j$ on run $G^*$ with $c(S)=c(S^*)$ there exists an isomorphism that
\begin{enumerate}
     \item respects edges,
     \item respects the \textit{initial} $h^1$-values,
     \item maps identical values between $h^{i'+1}(V(S))$ and $h^{j'+1}(V(S^*))$ to each other, and
     \item is a bijection $V(S) \rightarrow V(S^*)$.
\end{enumerate}

Since $h^1$-values are simply injective encodings of node labels, by proving this, we know the isomorphism will respect both edges and labels, and thus be a graph isomorphism.

% Note that $c$-enconding for a graph must be created at some step of the algorithm and does not change after it has been created.

\textit{Base Case}: $P(0,0)$. In this case $S_{1,2}, S^*_{1,2}$ are simply vertices, and $c(S_{1,2}) = c(S^*_{1,2})$ if they have the same $h^1$-values, which means they are isomorphic in terms of $h^1$-values and edges as well as bijective. Furthermore, the isomorphism maps same values between $h^1(V(S_{1,2}))$ and $h^1(V(S^*_{1,2}))$ to each other. 

\textit{Inductive Case}: $P(i,j) \rightarrow P(i+1,j)$.

So assume we at step $i+1 >0$ on $G$ have $S_{1,2}=S_1 \cup S_2 \cup (v_a,v_b)$, where $S_{1,2}$ is being encoded at step $i+1$. % and $A_{i+1}=(A_i - \{S_1,S_2\}) \cup \{S_{1,2}\}$.  

We need to prove that for any graph $S^*_{1,2}$ encoded at step $j' \leq j$ on run $G^*$ with $c(S_{1,2})=c(S^*_{1,2})$ we have a bijective graph isomorphism between $S_{1,2}$ and $S^*_{1,2}$ that respects the edges, initial $h^1$-values, and that maps identical values between $h^{i+2}(V(S_{1,2}))$ and $h^{j'+1}(V(S^*_{1,2}))$ to each other. The reason why we only need to focus on $S_{1,2}$ is because for all other graphs encoded at step $i'<i+1$ on $G$, their $c$-values and $h^{i'+1}$-values have not changed so they are covered by our inductive hypothesis $P(i,j)$.

%it is the only graph added to $A_i$ to construct $A_{i+1}$, and the $h^{i+1}$-values of all vertices of other graphs in $A_{i+1}$ are the same as the $h^{i}$-values and their $c$-values has not changed either. Thus, the case for all subgraphs in sets $A_{i'}, i'<i+1$ and all other subgraphs in $A_{i+1}$ already follows from inductive hypothesis $P(i,j)$.

Now we know that $|E(S^*_{1,2})|>0$ and $j'>0$ because $c(S_{1,2})$ does not include the special $zero$-symbol, and therefore, neither does $c(S^*_{1,2})$. Therefore, we can also write $S^*_{1,2}=S^*_1 \cup S^*_2 \cup (v^*_a,v^*_b)$ (specifically, $(v^*_a,v^*_b)$ is the edge used to encode $S^*_{1,2}$ from the encodings of $S^*_1$ and $S^*_2$). From Lemma \ref{lemma:disjoint} we know $S_1,S_2,S^*_1,S^*_2$ are connected graphs.
\begin{align*}
    c(S_{1,2})& = r(\{(c(S_1),h^{i+1}(v_a)), \\ &(c(S_2),h^{i+1}(v_b))\}, \\
    & \mathbbm{1}_{S_1 = S_2}) \\
    c(S^*_{1,2})& = r(\{(c(S^*_1),h^{j'}(v^*_a)),\\
    &(c(S^*_2),h^{j'}(v^*_b))\}, \\
    & \mathbbm{1}_{S^*_1 = S^*_2})
\end{align*}
By injectivity:
\begin{align*}
    &\big(\{(c(S_1),h^{i+1}(v_a)), (c(S_2),h^{i+1}(v_b))\}, \ \mathbbm{1}_{S_1 = S_2}\big) \\
    = &\big(\{(c(S^*_1),h^{j'}(v^*_a)), (c(S^*_2),h^{j'}(v^*_b))\}, \ \mathbbm{1}_{S^*_1 = S^*_2}\big)
\end{align*}
and we may assume without loss of generality that
\begin{align*}
(c(S_1),h^{i+1}(v_a))=(c(S^*_1),h^{j'}(v^*_a)) \\
(c(S_2),h^{i+1}(v_b))=(c(S^*_2),h^{j'}(v^*_b))
\end{align*}
else we can just relabel the graphs.

$S_1,S_2$ are encoded before step $i+1$ on $G$ (say steps $i_1$ and $i_2$ respectively) and $S^*_1,S^*_2$ are encoded before step $j'$ on $G^*$ (say steps ${j'}_1$ and ${j'}_2$ respectively). In addition, since $S_1,S_2 \in A_{i+1}$ their $h^{i_1+1}$ and $h^{i_2+1}$ values cannot have changed before step $i+1$ (because then they would have been removed already, see Lemma \ref{lemma:hvaluesnotchange}), so $h^{i+1}|_{V(S_1)}=h^{i_1+1}|_{V(S_1)}$ and $h^{i+1}|_{V(S_2)}=h^{i_2+1}|_{V(S_2)}$ (The same holds for $S^*_1, S^*_2$). Then, we have by our inductive hypothesis two bijective isomorphisms
\begin{align*}
    \phi_1:S_1 \rightarrow S^*_1, \ \ \ \phi_2:S_2 \rightarrow S^*_2
\end{align*}
with respect to edges and $h^{1}$-values, that maps identical values between $h^{i+1}(V(S_1))$ and $h^{j'}(V(S_1^*))$ (and between $h^{i+1}(V(S_2))$ and $h^{j'}(V(S_2^*))$) to each other, we must have
$$ \forall v \in S_1, \forall v^* \in S^*_1, h^{i+1}(v)=h^{j'}(v^*) \rightarrow \phi_1(v)=v^* $$
(and similarly for $\phi_2$).

% We know that 
% $$h^{i+2}(V(S_1))=r_v(c(S_{1}),h^{i+1}(V(S_1)))$$ and 
% $$h^{j'+1}(V(S^*_1))=r_v(c(S^*_{1}),h^{j'}(V(S^*_1)))$$
% Since $c:=c(S_1)=c(S^*_1)$ we also know since $r_v(c,\cdot)$ is injective across domains $h^{i+1}(V(S_1))$ and  $h^{j'}(V(S^*_1))$, that for any $v \in V(S_1), w \in V(S^*_1)$ with $h^{i+1}(v)=h^{j'}(w)$ we have $h^{i}(v)=h^{j'-1}(w)$. The same holds for $S_2$ and $S^*_2$.

% Specifically then, since 
% $$h^{i+1}(v_a)=h^{j'}(v^*_a), \ h^{i+1}(v_b)=h^{j'}(v^*_b)$$ we have 
% $$h^{i}(v_a)=h^{j'-1}(v^*_a), \ h^{i}(v_b)=h^{j'-1}(v^*_b)$$
% and thus
% $$ \phi_1(v_a)=v^*_a, \ \ \ \phi_2(v_b)=v^*_b $$

Specifically, since $h^{i+1}(v_a)=h^{j'}(v^*_a), h^{i+1}(v_b)=h^{j'}(v^*_b)$, we have
$$ \phi_1(v_a)=v^*_a, \ \ \ \phi_2(v_b)=v^*_b $$

Also we know that for all edges $(v_1,v_2) \in E(S_1)$, $(w_1,w_2) \in E(S_2)$ we have 
$$ (\phi_1(v_1), \phi_1(v_2)) \in E(S^*_1), \ \ \ (\phi_2(w_1),\phi_2(w_2)) \in E(S^*_2) $$
and the only new edge in $S_{1,2}$ is $(v_a,v_b)$, $v_a \in V(S_1), v_b \in V(S_2)$, and the only new edge in $S^*_{1,2}$ is $(v^*_a,v^*_b)$, $v^*_a \in V(s^*_1), v^*_b \in V(S^*_2)$.

Consider:
\begin{equation}
    \phi(v) =
    \begin{cases*}
      \phi_1(v) & if $v \in V(S_1)$ \\
      \phi_2(v)  & otherwise
    \end{cases*}
\end{equation}

% Since $\phi_1$ and $\phi_2$ respects $h^1$-values, so must $\phi$. We prove this by splitting into two cases:
We split into two cases:

\textit{Case 1}: ($ \mathbbm{1}_{S_1 = S_2} = False$). This implies that $S_1 \neq S_2$ and $S^*_1 \neq S^*_2$ (where $=$ is stronger than isormorphic). By Lemma \ref{lemma:disjoint} we have $V(S_1) \cap V(S_2)= V(S^*_1) \cap V(S^*_2) = \emptyset$. Since $\phi$ corresponds to a graph isomorphism on the disjoint $S_1\rightarrow S^*_1,S_2\rightarrow S^*_2$ and the new edge is respected, $\phi$ is a graph isomorphism between $S_{1,2}$ and $S^*_{1,2}$.

In addition, since $h^{i+2}$ and $h^{j'+1}$ are injective across domains $h^{i+1}(V(S_{1,2}))$ and $h^{j'}(V(S^*_{1,2}))$ it also means that $h^{i+2}$ and $h^{j'+1}$ are injective across domains $h^{i+1}(V(S_1))$ and $h^{j'}(V(S^*_1))$. Thus, if $h^{i+2}(v)=h^{j'+1}(w)$ with $v \in V(S_1), w \in V(S^*_1)$, then $h^{i+1}(v)=h^{j'}(w)$ such that by inductive hypothesis $\phi_1(v)=w$ and thus $\phi(v)=w$ (and similarly for $S_2,S^*_2$, and $\phi_2$). 

However, if there exists $v \in V(S_1), w \in V(S_2), u \in V(S^*_{1,2})$ with $h^{i+2}(v)=h^{i+2}(w)=h^{j'+1}(u)$ we need to make sure $\phi(v)=\phi(w)=u$ (to always map identical values to each other), but then $\phi$ would not be a graph isomorphism since $v\neq w$ (we know $S_1 \cap S_2 = \emptyset$). This could also be the case for $S^*_1,S^*_2,S_{1,2}$. But by uniqueness from $r_v$ we know $h^{i+2}(V(S_1)) \cap h^{i+2}(V(S_2)) = \emptyset$ and $h^{j'+1}(V(S^*_1)) \cap h^{j'+1}(V(S^*_2)) = \emptyset$, so this cannot happen, and we can conclude that identical values across $h^{i+2}(V(S_{1,2}))$ and $h^{j'+1}(V(S^*_{1,2}))$ are always mapped to each other.

\textit{Case 2}: ($ \mathbbm{1}_{S_1 = S_2} = True$). Which implies that $S_1 = S_2$ and $S^*_1 = S^*_2$ (in a stronger sense than isomorphic). This means $\phi = \phi_1$. Which means that $\phi$ is bijection (no new vertices are added, only an edge), and the new edge is also respected, so $\phi$ is a graph isomorphism between $S_{1,2}\rightarrow S^*_{1,2}$ that respects $h^1$-values and edges, because $\phi_1$ does so. 

In addition, $h^{i+2}$ and $h^{j'+1}$ are injective across domains $h^{i+1}(V(S_{1,2}))$ and $h^{j'}(V(S^*_{1,2}))$ with $V(S_{1,2})=V(S_1), V(S^*_{1,2})=V(S^*_1)$. Thus, if $h^{i+2}(v)=h^{j'+1}(w)$ with $v \in V(s_1), w \in V(S^*_1)$, then $h^{i+1}(v)=h^{j'}(w)$ such that by inductive hypothesis $\phi_1(v)=w$ and thus $\phi(v)=w$. Since $S_1 = S_2$ and $S^*_1 = S^*_2$ we can conclude that identical values across $h^{i+2}(V(S_{1,2}))$ and $h^{j'+1}(V(S^*_{1,2}))$ are mapped to each other.

By Lemma \ref{lemma:disjoint} we know these two cases are exhaustive. Thus, $\phi$ is a bijective isomorphism between $S_{1,2}$ and $S^*_{1,2}$ with respect to edges and $h^1$-values. Furthermore, the isomorphism maps identical values across $h^{i+2}(V(S_{1,2}))$ and $h^{j'+1}(V(S^*_{1,2}))$ to each other.

Since $h^1$-values are injective with respect to node labels, we are done.

% Since we assumed the $w$-values are an injective function of the $c$-values we are done.

\end{proof}

\subsection{Existence of Required Functions}

We start by proving that there exists no continuous injective function from $\mathbb{R}^2$ to $\mathbb{R}$.
\begin{theorem}
There exists no continuous injective function $f : \mathbb{R}^2 \to \mathbb{R}$.
\end{theorem}
\begin{proof}
Suppose $f:\mathbb{R}^2\rightarrow \mathbb{R}$ is continuous. Then the image (which is an interval in $\mathbb{R}^2$) of any connected set in $\mathbb{R}^2$ under $f$ is connected. Note that this is a non-degenerate interval (a degenerate interval is any set consisting of a single real number) since the function is injective. Now, if you remove a point from $\mathbb{R}^2$ it remains connected, but if we remove a point whose image is in the interior of the interval then the image cannot be still connected if the function is injective.
\end{proof}

We add some lemmas before we prove the main theorem of this section. All statements will be concerning NPA using the functions put forward in Section \ref{sec:reqfuncs}.

\begin{lemma}
\label{lemma:mhvaluesN0}
For NPA, all $m^1,m^2$ and $h$-values that appear are in $\mathbb{N}_0$
\end{lemma}
\begin{proof}
We show this through an informal induction argument. Since $h_{init}(v)=l(v) \in \mathbb{N}_+$ and $c_{init}(h)=(0,0,h+1)$ we know that all $h^1$-values are in $\mathbb{N}_0$, and for all $c$-values created at step $0$ we have $m^1,m^2 \in \mathbb{N}_0$. Now since new $m$-values are created from $m^1_1,m^2_1,m^1_2,m^2_2 \in \mathbb{N}_0$ through $m^1_{1,2} = m^2_1+m^2_2+1 \in \mathbb{N}_0, m^2_{1,2} = 2(m^2_1+m^2_2+1) \in \mathbb{N}_0$ it is not hard to see that all $m^1,m^2$ that appear will be in $\mathbb{N}_0$. Similarly, new $h^{i+1}$-values are created from $h^{i}$-values through $h^{i+1}(v)=h^{i}(v) \in \mathbb{N}_0$ or $h^{i+1}(v)=h^{i}(v) + m^1 \in \mathbb{N}_0$ (since $m^1 \in \mathbb{N}_0$), so all $h$-values will be in $\mathbb{N}_0$.
\end{proof}

\begin{lemma}
For any graph $s_k$ encoded by Algorithm \ref{alg:NPA} at step $i$ on run $G$ we have $m^2_k > \max(h^{i+1}(V(S_k)))$ and each value in $h^{i+1}(V(S_k)) = r_v(c(S_k),h^{i}(V(S_k)))$ is unique.
\label{lemma:preweak}
\end{lemma}

\begin{proof}
We will prove this by strong induction on the number of steps $i$ of the algorithm on run $G$. Property $P(i)$ is that any graph $S_k$ encoded at step $i$ on run $G$:
\begin{itemize}
    \item $m^2_k > \max(h^{i+1}(V(S_k)))$, and
    \item each value in $h^{i+1}(V(S_k)) = r_v(c(S_k),h^{i}(V(S_k)))$ is unique
\end{itemize}

\textit{Base Case}: $P(0)$. This means $S_k$ consists of a single vertex $v$. Thus, $h^1(V(S_k))=\{ l(v) \} \subset \mathbb{N}_+$ and it is unique. Consequently, $m_k^2=l(v)+1 > 0$, such that $m_k^2 > \max(h^1(V(S_k))=l(v)$. We also note that $m^1_k = 0$.
%Trivially, $r_v(c(s_k),h^1(V(S_k)))$ is unique, since it is not evaluated.

\textit{Inductive Case}: $ (\forall i' \leq i, P(i')) \rightarrow P(i+1)$.

Since $i+1>0$ we have $|E(s_k)|>0$ so we can write $V(S_{1,2}):=V(S_k)=V(S_1)\cup V(S_2)$, where $S_1,S_2$ were encoded before step $i+1$, say step $i_1$ and $i_2$ respectively. By inductive hypothesis, this means that all values in $h^{i_1+1}(V(S_1))$ and all values in $h^{i_2+1}(V(S_2))$ are unique, and since $S_1,S_2 \in A_{i+1}$, by Lemma \ref{lemma:hvaluesnotchange}, these $h$-values cannot have changed before step $i+1$ (i.e. $h^{i_1+1}|_{V(S_1)}=h^{i+1}|_{V(S_1)}, h^{i_2+1}|_{V(S_2)}=h^{i+1}|_{V(S_2)}$). Thus, each value in $h^{i+1}(V(S_1))$ and each value in $h^{i+1}(V(S_2))$ is unique. By injective hypothesis we also know that 
$$m^2_1 > \max(h^{i+1}(V(S_1))), \ m^2_2 > \max(h^{i+1}(V(S_2)))$$

From Lemma \ref{lemma:mhvaluesN0}, we know $m^2_1,m^2_2 \in \mathbb{N}_0$ and all $h$-values in $\mathbb{N}_0$, i.e. they are non-negative.

Now we have, with $m^1_{1,2}=m^2_1+m^2_2+1 > 0$, that
\begin{align*}
h^{i+2}&(V(S_{1,2})):=r_v(c(S_{1,2}),h^{i+1}(v)) = \\
&\left\{\begin{array}{lr}
        h^{i+1}(v)+m^1_{1,2}, \ \ & \text{if } v \in V(S_1) \\
        h^{i+1}(v), \ \ & \text{else} 
        \end{array} \right\}
\end{align*}
This means now that each value in $h^{i+2}(V(S_1))$ and each value in $h^{i+2}(V(S_2))$ is unique. This is easier to see for $h^{i+2}(V(S_1))$ because $r_v$ is an injective function on the values of $h^{i+1}(V(S_1))$ which we know are all unique. However, since 
$$m^1_{1,2} > \max(h^{i+1}(V(S_2))), \ \min(h^{i+1}(V(S_1))) \geq 0$$
$r_v$ is also injective on $h^{i+1}(V(S_2))$. To prove this, suppose $r_v(h^{i+1}(v))=r_v(h^{i+1}(w))$ with $v,w\in V(S_2)$, then $h^{i+1}(v)=h^{i+1}(w)$ unless, w.l.o.g, $v \in V(S_1), w \notin V(S_1)$ from which we reach a contradiction since $\min(h^{i+1}(V(S_1))) + m^1_{1,2} > \max(h^{i+1}(V(S_2)))$. 

Since $\max(h^{i+1}(V(S_1))) + m_2^2 + 1 > \max(h^{i+1}(V(S_2)))$ we have
\begin{align*}
\max(h^{i+2}(V(S_{1,2}))) =& \max(h^{i+1}(V(S_1))) \\
&+ m^2_1+m^2_2 + 1 \\ 
&< 2m^2_1 + m^2_2 + 1  
\end{align*}
Since $m^2_{1,2}=2m^2_1+2m^2_2+2 > 0$ this means that $\max(h^{i+2}(V(S_{1,2}))) < m^2_{1,2}$. We can also conclude $m^1_{1,2}, m^2_{1,2} \in \mathbb{N}_+$.

By Lemma \ref{lemma:disjoint}, we know that either $S_1 = S_2$ or $S_1 \cap S_2 = \emptyset$. If $S_1 = S_2$, then $V(S_{1,2})=V(S_1)=V(S_2)$ such that $h^{i+2}|_{V(S_{1,2})}=h^{i+1}|_{V(S_1)} + m^1_{1,2}$, which means that each value in $h^{i+2}(V(S_{1,2}))$ is unique because each value in $h^{i+1}(V(S_1))$ is unique. Thus we are done, and we now assume that $S_1 \cap S_2 = \emptyset$.

This means that $V(S_1) \cap V(S_2) = \emptyset$ and 
$$h^{i+2}(V(S_1)) \cap h^{i+2}(V(S_2)) = \emptyset$$
since $m^1_{1,2} > m^2_1 + m^2_2 > \max(h^{i+1}(V(S_2)))$, $\max(h^{i+2}(V(S_2))) = \max(h^{i+1}(V(S_2)))$. Thus, all values in
$$h^{i+2}(V(S_{1,2})) = h^{i+2}(V(S_1)) \sqcup h^{i+2}(V(S_2))$$
are unique. 

Thus we have proved $P(i+1)$.
\end{proof}

\begin{corollary}
This also means that $m^1_k=0$ if and only if $|E(S_k)|=0$ (i.e. in the base case). Thus, it serves as the required $zero$-symbol.
\end{corollary}

% \begin{lemma}
% \label{lemma:mhvaluesN0}
% All $h$-values are in $\mathbb{N}_0$.
% \end{lemma}
% \begin{proof}
% Since $h_{init}(l(v))=l(v) \in \mathbb{N}_+$ and all $m^1,m^2 \in \mathbb{N}_0$, we have $h^{i+1}(v) = r_v(c, h^{i}) \in \mathbb{N}_0, \forall i \in \mathbb{N}_0$; the lemma follows.
% \end{proof}

Armed with this lemma we will now prove the following:

\begin{lemma}
For all graphs $S,S^*$ encoded at step $i$ run $G$ and $j$ run $G^*$ respectively with $c:=c(S)=c(S^*)$, $r_v(c,\cdot)$ is injective across domains $h^{i}(V(S))$ and $h^{j}(V(S^*))$.
\end{lemma}

\begin{remark}
We reiterate, with a function $f:X \rightarrow Y$ being injective across domain $X_1$ and $X_2$ with $X_1, X_2 \subset X$, we mean that for all $x_1 \in X_1, x_2 \in X_2$ with $f(x_1)=f(x_2)$ we have $x_1=x_2$.
\end{remark}

\begin{proof}
% We prove this by double induction on $i$ and $j$. No, we don't

First if $i=0$ or $j=0$ we know that both $i=j=0$ due to the $zero$-symbol, and then it is vacuously true, because $h^0$ does not exist and $r_v$ is not applied. So we assume $i,j>0$.

Since $i,j>0$ we have $V(S)=V(S_1)\cup V(S_2)$, $V(S^*)=V(S^*_1)\cup V(S^*_2)$. We also know $(m^1,m^2)=(m^1_{*},m^2_{*})$. By Lemma \ref{lemma:disjoint} we know that either $S_1 = S_2$ or $S_1 \cap S_2 = \emptyset$.

If $S_1 = S_2$, then since $c(S)=c(S^*)$ we also have $S^*_1 = S^*_2$, which means that $V(S)=V(S_1)=V(S_2)$ and $V(S^*)=V(S^*_1)=V(S^*_2)$. This means that $r_v(c,h)=h+m^1=h+m^1_*$, which then is injective and in particular injective across $h^i(V(S))$ and $h^j(V(S^*))$. Thus, we now assume that $S_1 \cap S_2 = \emptyset$.

% This means that $V(S_1)\cap V(S_2) = \emptyset$ and by uniquness (Lemma \ref{lemma:preweak}) of values in $h^{i+1}(V(S))$ we have that 
% $$h^{i+1}(V(S_1)) \cap h^{i+1}(V(S_2)) = \emptyset$$
% (and same holds for $S^*_1,S^*_2$).

This means that $V(S_1)\cap V(S_2) = \emptyset$. Now suppose
$$r_v(c,h^{i}_a)=r_v(c,h^{j}_b)$$
with $h^{i}_a \in h^{i}(V(S)), h^{j}_b \in h^{j}(V(S^*))$. Consider two cases:

\textit{Case 1}: $h^{i}_a \in h^{i}(V(S_1))$. Then $$r_v(c,h^{i}_a)=h^{i}_a+m^1=h^{i}_a+m^1_{*}$$
Since $m^1_{*} > \max(h^{j}(V(S^*_2))) \geq 0$ and $h^{i}_a \geq 0$ (Lemma \ref{lemma:preweak} and \ref{lemma:mhvaluesN0}) we must have $h^{j}_b \in h^{j}(V(S^*_1))$ such that $$r_v(c,h_b^{j})=h_b^{j}+m^1_{*}$$
Because else
$$r_v(c,h^{j}_b)=h^{j}_b < m^1_* < r_v(c,h^{i}_a)$$
This implies that $h^{i}_a=h^{j}_b$.

\textit{Case 2}: $h^{i}_a \notin h^{i}(V(S_1))$ which means that $h^{i}_a \in h^{i}(V(S_2))$. Suppose by contradiction that $h^j_b \in h^j(V(S^*_1))$ then
$$ r_v(c,h^i_a) = h^i_a = r_v(c,h^i_b) = h^i_b + m^1_* = h^i_b + m^1 $$
But since $m^1 > \max(h^i(V(S_2)) \geq 0$ and $h^i_b \geq 0$ (Lemma \ref{lemma:preweak} and \ref{lemma:mhvaluesN0}) we get a contradiction. This means $h^{j}_b \notin h^{j}(V(S^*_1)), h^{j}_b \in h^{j}(V(S^*_2))$ such that $$r_v(c,h^{i}_a)=h^{i}_a=r_v(c,h^{j}_b)=h^{j}_b$$

We are done.
\end{proof}

Consider the following functions:
\begin{align*}
\tau(i,j)= \frac{(i+j)(i+j+1)}{2}+j, \ \
\rho(i,j)=(i+j,ij)
\end{align*}

\begin{lemma} 
\label{lemma:threeclaims}
Two claims:
\begin{itemize}
    \item $\tau:\mathbb{R}\times\mathbb{R} \rightarrow \mathbb{R}$ is continuous and injective in $\mathbb{N} \times \mathbb{N} \rightarrow \mathbb{N}$.
    \item $\rho:\mathbb{R}\times \mathbb{R} \rightarrow \mathbb{R}$ is continuous and injective in $\{\{i,j\} \ | \ i,j \in \mathbb{N} \} \rightarrow \mathbb{N}^2$.

\end{itemize}
\end{lemma}
\begin{proof}
$\tau$ is the well-known Cantor Pairing Function, see for example \href{https://en.wikipedia.org/wiki/Pairing_function}{Wikipedia} for proof of its bijective properties on $\mathbb{N}^2 \rightarrow \mathbb{N}$, it is clearly continuous on $\mathbb{R}^2 \rightarrow \mathbb{R}$.

$\rho$ is cleary continuous in $\mathbb{R}^2 \rightarrow \mathbb{R}^2$ and if $i,j \in \mathbb{N}$ then $\rho(i,j) \in \mathbb{N}^2$. We will prove that it is injective in $\{\{i,j\} \ | \ i,j \in \mathbb{N} \} \rightarrow \mathbb{N}^2$: 

Suppose $(i+j,ij)=(x,y)$ we want to express $i$ and $j$ in terms of $x$ and $y$. Rearranging and substituting, we get $i=x-j \Rightarrow (x-j)j=y \Rightarrow j^2-xj+y = 0$. Using the quadratic formula, and by symmetry, we get
$$ j = \frac{x \pm \sqrt{x^2 - 4y}}{2}, \ \ i = \frac{x \pm \sqrt{x^2 - 4y}}{2} $$
If $j=\frac{x + \sqrt{x^2 - 4y}}{2}, i = \frac{x - \sqrt{x^2 - 4y}}{2}$ (or other way around) the conditons $i+j=x,ij=y$ holds. But if $j=\frac{x + \sqrt{x^2 - 4y}}{2} = i = \frac{x + \sqrt{x^2 - 4y}}{2}$ then $i+j = x+\sqrt{x^2-4y}$ and $ij=\frac{x^2}{4} + x\sqrt{x^2-4y} + \frac{x^2-4y}{4}$ and conditions hold iff $x^2=4y$ which takes us back to our previous case. Similarly, if $j=\frac{x - \sqrt{x^2 - 4y}}{2} = i = \frac{x - \sqrt{x^2 - 4y}}{2}$ then $i+j = x-\sqrt{x^2-4y}, ij = \frac{x^2}{4} - x\sqrt{x^2-4y} - \frac{x^2-4y}{4}$ and conditions hold iff $x^2=4y$ which again takes us back to our first case. Thus, we have proved that $\rho$ is injective.

\end{proof}

\begin{lemma}
In the above setup, there exists a continuous and bounded function $r:\mathbb{R}^9 \rightarrow \mathbb{R}$ that is injective in $\{ \mathbb{N}^4, \mathbb{N}^4 \}\times \mathbb{N}$. Namely,
\begin{align*}
r(&y_1,h_1,m_1,n_1,y_2,h_2,m_2,n_2,b) = \tau\big(\tau\big(\rho(\tau^4(y_1,h_1,m_1,n_1), \tau^4(y_2,h_2,m_2,n_2))\big), b\big)
\end{align*}
\label{lemma:injective_r}
\end{lemma}

%We prove Lemma \ref{lemma:injective_r}
\begin{proof} %[Proof of Lemma \ref{lemma:injective_r}]
The proof follows from Lemma \ref{lemma:threeclaims}.
%$r$ is the following function:
%\begin{align*}
%r(&y_1,h_1,m_1,n_1,y_2,h_2,m_2,n_2,b) = \\
%& \tau\big(\tau\big(\rho(\tau^4(y_1,h_1,m_1,n_1), %\tau^4(y_2,h_2,m_2,n_2))\big), b\big)
%\end{align*}
%Thus, the proof follows from Lemma \ref{lemma:threeclaims}.
\end{proof}

\begin{lemma}
\label{lemma:reqfuncsintegers}
For the functions defined in Section \ref{sec:reqfuncs} and in this section, when used in NPA, we always have (i) $h^j(v) \in \mathbb{N}_0$ and (ii) $c(S_k)=(y_k, m^1_k, m^2_k) \in \mathbb{N}_0 \times \mathbb{N}_0 \times \mathbb{N}_0 = \mathbb{N}_0^3$.
\end{lemma}
\begin{proof}
(i) $h^j(v) \in \mathbb{N}_0$ follows immediately from Lemma \ref{lemma:mhvaluesN0}. Note that (ii) is true for all $c$-values encoded at step $0$ in NPA via $c_{init}$ since all $h$-values are in $\mathbb{N}_0$, also we know that all $m_k^1,m_k^2 \in \mathbb{N}_0$ from Lemma \ref{lemma:mhvaluesN0}. Thus, the only thing we need to consider is the subsequent application of $r$, and it is applied to $h$-values, $c$-values, and $\{0,1\}$-indicators, all of which are in $\mathbb{N}_0$, to create new $c$-values. Since $r$ takes $(\mathbb{N}_0)^*$ to $(\mathbb{N}_0)^*$, which can be seen by inspection, the lemma follows.
\end{proof}

\begin{lemma}
The $r_c$ function with the $r$-function from Lemma \ref{lemma:injective_r} is injective in all its variables.
\end{lemma}
\begin{proof}
Suppose 
$$ r_c(\{(c^1_1,h^1_1),(c^1_2,h^1_2)\},b_1) = r_c(\{(c^2_1,h^2_1),(c^2_2,h^2_2)\},b_2)  $$
Where 
\begin{align*}
c^1_1 &= ( y^1_1,m^1_1,n^1_1), \ c^1_2 = ( y^1_2,m^1_2,n^1_2)  \\ 
c^2_1 &= ( y^2_1,m^2_1,n^2_1), \ c^2_2 = ( y^2_2,m^2_2,n^2_2)
\end{align*}
This means that 
\begin{align*}
\big(r&(y^1_1, h^1_1, m^1_1, n^1_1, y^1_2, h^1_2, m^1_2, n^1_2, b_1), \\
&n^1_1+n^1_2+1, \ 2n^1_1+2n^1_1+2 \big) = \\
&\big( r(y^2_1, h^2_1, m^2_1, n^2_1, y^2_2, h^2_2, m^2_2, n^2_2, b_2), \\
&n^2_1+n^2_2+1, \ 2n^2_1+2n^2_2+2 \big)
\end{align*}
Thus, from Lemma \ref{lemma:injective_r} we know $r$ is injective in 
$\{ \mathbb{N}^4, \mathbb{N}^4 \}\times \mathbb{N}$. 
%Since $r$ maps $\mathbb{N}^9$ to $\mathbb{N}$, is applied recursively, and $y_{init} \in \mathbb{N}$ we know all input to $r$ are integers. 
By Lemma \ref{lemma:reqfuncsintegers} we know all input to $r$ are in $\mathbb{N}_0$, thus, $r$ is injective, which gives us 
$$ \big( \{(c^1_1,h^1_1),(c^1_2,h^1_2)\},\ b_1 \big) = \big(\{(c^2_1,h^2_1),(c^2_2,h^2_2)\},\ b_2 \big) $$
and we are done.
\end{proof}

\begin{lemma}
For Algorithm \ref{alg:NPA} there exists functions $r_{v}$, $r_c$, $h_{init}, c_{init}$ that satisfies the requirements put forward in Theorem \ref{thm:main}.
\label{lemma:exist}
\end{lemma}

\begin{proof} %[Proof of Lemma \ref{lemma:exist}]
Consider the functions defined in Section \ref{sec:reqfuncs} and in this section, as well as the results. The lemma follows.
\end{proof}

% \begin{corollary}
% We only need to apply a function $r_{v}$ to the $h$-values of $V(s_1)$ (since the other function can be the identity function) and only where the $h$-values are identical, i.e. $\forall v_1 \in V(s_1), \forall v_2 \in V(s_2), h(v_1)=h(v_2)$. 
% \end{corollary}

\subsection{Corollaries}

We add a remark about the subgraphs that are encoded during runs of NPA on a graph $G$.

\begin{remark}
On one run of NPA on graph $G$, the multiset $W(G)$ encodes a collection of subgraphs of $G$, for example, these subgraphs always include the vertices and the largest (by inclusion) connected subgraphs. The order in which edges are processed determines which other subgraphs that are encoded, but it is not too hard to see that if NPA is run on all possible orders on edges, and without NPA changing the order, it will encode each combination of disjoint connected subgraphs. Since any subraph consists of a collection of disjoint connected subgraphs, it will indirectly encode all possible subgraphs. 
%Now, the encodings can also depend on the order of nodes within edges, but if NPA is run on every possible indexing of the nodes of $G$, i.e. on all $H \in [G]$, then it will encode all graphs $S \subset H, H \in [G]$, which is the multiset-set of encodings $\bigcup_{H \in [G]}W(H)$. %=\bigcup_{W(H) \in \bm{W}([G])}W(H)$.
\end{remark}

Full proof of Lemma \ref{lemma:fcountinjective}
\begin{proof}
\label{proof:lemmapow}
% Mapping $Z : \mathcal{X} \rightarrow \mathbb{N}$, and $|X| < N$. Consider $f(x) = N^{-Z(x)}$. See Appendix.
(From \cite{powerful}). We first prove that there exists a mapping $f$ so that $\sum_{x \in X} f (x)$ is unique for each multiset $X$ bounded size. Because $\mathcal{X}$ is countable, there exists a mapping $Z : \mathcal{X} \rightarrow \mathbb{N}$ from $x \in \mathcal{X}$ to natural numbers. Because the cardinality of multisets $X$ is bounded, there exists a number $N \in \mathbb{N}$ so that $|X| < N$ for all $X$. Then an example of such $f$ is $f(x) = N^{-Z(x)}$. This $f$ can be viewed as a more compressed form of an one-hot vector or $N$-digit presentation. Thus, $h(X) = \sum_{x \in X} f(x)$ is an
injective function of multisets. $\phi( \sum_{x \in X} f (x))$ is permutation invariant so it is a well-defined multiset function. For any multiset function $g$, we can construct such $\phi$ by letting $\phi(\sum_{x \in X}f(x))= g(X)$. Note that such $\phi$ is
well-defined because $h(X) = \sum_{x \in X} f(x)$ is injective.
\end{proof}

\begin{corollary}
There exists a function $f$ such that any two graphs $G$ and $H$ in $\mathcal{G}_b$ are isomorphic if $ \sum_{w \in W(G)}f(w) = \sum_{w \in W(H)}f(w)$.
\label{cor:sumbound}
\end{corollary}

\begin{remark}
\label{remark:correctionremark}
%For a $[G]$ there is a finite number of encodings $W(H), H \in [G]$ where each such $W(H)$ is a finite multiset of encodings corresponding to subgraphs of $H$. 
Given a graph isomorphism class $[S]$ and assuming NPA does not change the order of the edges, there is a Turing-decidable function $f_{[S]}:\mathcal{G} \rightarrow [0,1]$ that on input $G$ returns $1$ if there exists $S \in [S], H \in [G]$ with $S \subset H$ and $0$ otherwise; in pseudo-code:
\begin{align*}
&f_{[S]} \text{ on input } G, \\
&\hspace{0.5cm} \forall H \in [G], \forall S \in [S], \\
&\hspace{1.0cm} \text{if }W(S) \subset W(H) \text{ return } 1, \\
&\hspace{0.5cm} \text{return }0
\end{align*}
which is Turing-decidable since for any $G\in \mathcal{G}$ all such sets $[G], [S], W(H), W(S)$ are finite. However, a similar function for detecting the presence of a subgraph in isomorphism class $[S]$ in graph $G$ given we only have one encoding $E(G)$ for all of $G$ must not exist. Without some subset-information in the encoding we are left to (pseudo-code):
\begin{align*}
&f_{[S]} \text{ on input } G, \\
&\hspace{0.5cm}\forall H \in \mathcal{G}, \exists S \in [S], S \subset H, \\
&\hspace{1.0cm}\text{if }E(G) = E(H)\text{ return } 1, \\
&\hspace{0.5cm}\text{return }0
\end{align*}
which is Turing-recognizable but not Turing-decidable, because the number of graphs $H \in \mathcal{G}$ that contain subgraphs in $[S]$ is infinite. This points to the strength of having the encoding of a graph $G$ coupled with encodings of its subgraphs.
\end{remark}

\subsection{Use of Neural Networks}

We make use of the following functions:
\begin{align*}
    c&_{init}(i) = (0,0,i+1)\\ 
    f&_1(i,j) = i+j+1 \\
    f&_2(i,j) = 2i+2j+2 \\
    r&(y_1,h_1,m_1,n_1,y_2,h_2,m_2,n_2,b) = \\
& \tau\big(\tau\big(\rho(\tau^4(y_1,h_1,m_1,n_1), \tau^4(y_2,h_2,m_2,n_2))\big), b\big) \nonumber \\
r&_v(\dots,m,h,\mathbbm{1}_{ind}) = h+\mathbbm{1}_{ind}m 
\end{align*}
Where
\begin{align*}
\tau(i,j)= \frac{(i+j)(i+j+1)}{2}+j, \ \
\rho(i,j)=(i+j,ij)
\end{align*}

To a lesser extent we use
$$f_3(i) = N^{-i}$$

By Theorem \ref{thm:finiteuniversal}, NNs can perfectly approximate any function on a finite domain so the case of $\mathcal{G}_b$ is straightforward. 
% We know from Theorem [] that NNs can perfectly approximate any function on a finite domain. Since the set of bounded graphs $\mathcal{G}_b$ is finite, we know from previous results that NNs can be used to perfectly approximate any functions on $\mathcal{G}_b$, including $c_{init}$, $f_1$, $f_2$, $r$, $r_v$, and $f_3$.
% \begin{remark}
% For unbounded graphs, we will assume that a bounded injective function like Sigmoid or Tanh can be pointwise approximated by a NN. This is because they are often used as activation functions and help us bound our domains.
% \end{remark}
However, for countably infinite $\mathcal{G}$ the situation is different. %Consider functions from Section \ref{sec:reqfuncs} and \ref{sec:corollaries} (Lemma \ref{lemma:fcountinjective}).
Note that these functions are continuous (in $\mathbb{R}^*$) but not bounded and that we are applying these functions recursively and would want both the domain and the image to be bounded iteratively.
%(Also note $f_3:\mathbb{R} \rightarrow \mathbb{R}^+$ is injective, since the exponential is injective with $f_3^{-1}:\mathbb{R}^+ \rightarrow \mathbb{R}, f_3^{-1}(x) =-\ln(x)$)
Without losing any required properties we can compose these functions, $f$, with an injective, bounded, and continuous function with continuous inverse such as Sigmoid, $\sigma$, in the following way $f^*=\sigma \circ f \circ \sigma^{-1}$, and use $h_{init}(l(v)) = \sigma(l(v))$. Then these functions can be pointwise approximated by NNs. 

\begin{lemma}
$\sigma:\mathbb{R}\rightarrow (0,1)$, $\sigma(x)=\frac{1}{1+e^x}$ is continuous, bounded, and injective. Also, its inverse $\sigma^{-1}:(0,1)\rightarrow \mathbb{R}$ is continuous and injective. 
\end{lemma}
\begin{proof}
$\sigma$ is continuous since the exponential function is continuous, and it is clearly bounded with $\operatorname{im}(\sigma) = (0,1)$. Furthermore, its inverse is $\sigma^{-1}(x) = \ln(\frac{1-x}{x}) : (0,1) \rightarrow \mathbb{R}$, thus it is injective. Since $\ln$ is continuous so is $\sigma^{-1}$, and since $\sigma^{-1}$ is the inverse of a function, it is injective.
\end{proof}
The required functions then become:
\begin{align*}
&c^*_{init}:(0,1) \rightarrow (0,1), \ c^*_{init}=\sigma \circ c_{init} \circ \sigma^{-1}\\
&f^*_1:(0,1)^2\rightarrow (0,1), \ f^*_1=\sigma \circ f_1 \circ \sigma^{-1}\\
&f^*_2:(0,1)^2\rightarrow (0,1), \ f^*_2=\sigma \circ f_2 \circ \sigma^{-1} \\
&r^*:\{(0,1)^4,(0,1)^4\}\times (0,1) \rightarrow (0,1), \ r^*=\sigma \circ r \circ \sigma^{-1} \\
&r_v^*: (0,1)^3 \rightarrow (0,1), \ r_v^*=\sigma \circ r_v \circ \sigma^{-1} 
% &f_3^*:(0,1) \rightarrow (0,1), \ f_3^* = \sigma \circ f_3 \circ \sigma^{-1} 
\end{align*}

It follows from the setup and Lemma \ref{lemma:reqfuncsintegers} that if $\operatorname{im}(h_{init}) \subset \{ \sigma(i) \ | \ i \in \mathbb{N}\} $ then all these functions maintain their required properties. All these functions are continuous and bounded (iteratively on $(0,1)$ by $(0,1)$) in $\mathbb{R}^*$. Thus, by Theorem \ref{thm:2ndapproxthm}, they can be pointwise approximated by a NN. Yet, for $f_3$ the situation is a little different because we care about the sum $\sum_{x \in X} f_3(x)$ over a bounded multiset $X$. However, note that all the domain consists of $\mathbb{N}_0$ so $f_3$ is bounded by $(0,1]$. Thus we can pointwise approximate 
$$f^*_3:(0,1)\rightarrow (0,1]: f^*_3 = f_3 \circ \sigma^{-1}$$
which suffices, and if $X$ is bounded, so is the sum. 

However, it also follows, due to the use of $\sigma$, that the pointwise approximation error is going to be more likely to cause problems for large values.

% They are injective because each domain $(0,1)=\sigma(\mathbb{N})$, and for all $c$-values we have $c \in (0,1)^3$. 
% \begin{theorem}
% Functions $r_v, r_c, h_{init}, c_{init}$ in Algorithm \ref{alg:NPA} that satisfies requirements of Theorem \ref{thm:main}, and function $f_3(i)=N^{-i}$ from Section \ref{sec:corollaries}, can be perfectly approximated by NNs for graphs in $\mathcal{G}_b$ and pointwise approximated for graphs in $\mathcal{G}$.
% \end{theorem}

% \begin{proof}
% Proof follows from discussion above.
% \end{proof}
\subsubsection{Approximation Error and its Accumulation}

Recursive application of a NN might increase the approximation error. We have the following equations describing successive compositions of a NN $\varphi$:
\begin{align*}
    &||f(f(x))-\varphi(\varphi(x))|| \\ 
    &= ||f(f(x)) - \varphi(f(x) + \epsilon)|| \\
    &= ||f(f(x)) - f(f(x)+\epsilon)+ \epsilon||
\end{align*}
Future work should investigate the effects of this likely accumulation.

\subsection{Class-Redundancy, Sorting, Parallelize, and Subgraph Dropout}

Again, the class-redundancy in the algorithm and functions we propose enters at the sort functions $s_e$ (sorts edges) and $s_v$ (sorts nodes within edges). 
%Since identical ordering of the edges and the nodes within edges give the same encoding, 
Thus, a loose upper bound on the \textit{class-redundancy} is $O((m!)2^m)$.
%Alternatively, we can simply use an ordering on the nodes $0(n!)$ if it is smaller. 
However, a more exact upper bound is $O((t_1!)(t_2!)\dots(t_k!)(2^p))$, where $t_i$ are the sizes of the consecutive ties for the sorted edges, and $p$ (bounded by $m$) is the number of ties for the sorting of nodes within edges. 
An even better upper bound is $$O((t_{1,1}!)\dots(t_{1,l_1}!)(t_{2,1}!)\dots(t_{k,l_k}!)(2^p))$$
where each $t_{i,j}$ is the number of ties within group $j$ of groups of subgraphs that could be connected within the tie $i$. The order in between disconnected tied subgraph groups does not affect the output. 
%Again, $p$ is the number of edges where $s_v$ has ties .
% See Appendix for \#edge-orders, i.e. $O((t_{1,1}!)\dots(t_{1,l_1}!)(t_{2,1}!)\dots(t_{k,l_k}!))$, on some datasets. 

In Table \ref{table:redundancy} you can find \#edge-orders, that is $O((t_{1,1}!)\dots(t_{1,l_1}!)(t_{2,1}!)\dots(t_{k,l_k}!))$, and \#levels on some datasets.

\begin{table*}%[h]
\caption{Edge-orders and levels.}
\begin{center}
\begin{tabular}{l l l l l l}
\hline
Datasets: &  & NCI1 & MUTAG & PROTEINS & PTC  \\
Avg \# nodes: & & 30 & 18 & 39 & 26 \\
Avg \# edges: & & 32 & 20 & 74 & 26 \\
\hline
%$O($avg \# edge-orders$)$: & 10^{46} & 10^8 & 10^{143} & 10^{69} \\
%$O($avg \# class-redundancy$)$: & 10^{54} & 10^{13} & 10^{211} & 10^{79} \\
% $O($avg \# edge-orders$)$: & 10^{30} & 10^9 & 10^{143} & 10^{25} \\
% $O($avg \# class-redundancy$)$: & 10^{39} & 10^{14} & 10^{211} & 10^{37} \\
$O($median \# edge-orders$)$: & degs-and-labels & $10^7$ & $10^5$ & $10^{13}$ & $10^5$ \\
$O($median \# edge-orders$)$: & two-degs & $10^9$ & $10^7$ & $10^{23}$ & $10^6$ \\
$O($median \# edge-orders$)$: & one-deg & $10^{20}$ & $10^{14}$ & $10^{36}$ & $10^{16}$ \\
$O($median \# edge-orders$)$: & none & $10^{31}$ & $10^{17}$ & $10^{62}$ & $10^{23}$ \\
\hline
Avg samples \# levels: & degs-and-labels & 12 & 11 & 41 & 9 \\
Avg samples \# levels: & two-degs & 12 & 10 & 41 & 9 \\
Avg samples \# levels: & one-deg & 14 & 11 & 41 & 13 \\
Avg samples \# levels: & none & 12 & 14 & 39 & 13 \\
%O($median \# class-redundancy$)$: & $all$ & 4096 & 256 & 10^{6} & 128 \\

\hline
\end{tabular}
\end{center}

\label{table:redundancy}
\end{table*}

\subsection{Neural Networks}

% We use NNs for all non-sort functions. For NPA and NPBA we use for $r_c$ a tree-LSTM similar to those in \cite{treelstm} and for $r_v$ we use a LSTM. See Appendix for details. 

For NPBA we let $c(S_i) = (c^1_i,c^2_i)$ be the encoding for a subgraph $S_i$ and use for $r_c$:
\begin{align*}
    i &= \sigma(W_{i} (c^2_0 + c^2_1) + b_{i}) \\
    f_1 &= \sigma(W_{f}c^2_0 + b_{f}) \\
    f_2 &= \sigma(W_{f}c^2_1 + b_{f}) \\
    g &= \tanh(W_{g} (c^2_0 + c^2_1) + b_{g}) \\
    o &= \sigma(W_{o} (c^2_0 + c^2_1) + b_{o}) \\
    c^1_{1,2} &= f_1 * c^1_{0} + f_2 * c^1_{1} + i * g \\
    c^2_{1,2} &= o * \tanh(c^1_{1,2}) 
\end{align*} 

For the NPA we use for $r_c(\{(c(S_1),h_1), (c(S_2),h_2)\}, s:=\mathbbm{1}_{S_1 = S_2})$:
\begin{align*}
    i &= \sigma(W_{i,h}(h_1+h_2) + W_{i,c} (c^2_1 + c^2_2) + W_{i,s}s + b_{i}) \\
    f_1 &= \sigma(W_{f,h}h_1 + W_{f,c}c^2_1 +  W_{f,s}s + b_{f}) \\
    f_2 &= \sigma(W_{f,h}h_2 + W_{f,c}c^2_2 + W_{f,s}s + b_{f}) \\
    g &= \tanh(W_{g,h}(h_1+h_2) + W_{g,c} (c^2_1 + c^2_2) + W_{g,s}s + b_{g}) \\
    o &= \sigma(W_{o,h}(h_1+h_2) + W_{o,c} (c^2_1 + c^2_2) + W_{o,s}s + b_{o}) \\
    c^1_{1,2} &= f_1 * c^1_{1} + f_2 * c^1_{2} + i * g \\
    c^2_{1,2} &= o * \tanh(c^1_{1,2}) 
\end{align*} 
Where $s= \mathbbm{1}_{S_1 = S_2}$ and the encoding for a subgraph $S_i$ is $c(S_i) = (c^1_i, c^2_i)$ and the $h$-value of a node $v_j$ is encoded by $h_j$ (so $h_1$ and $h_2$ above encode $h(v_a)$ and $h(v_b)$ respectively).

For $r_{v}(c(S_{1,2}), h_v, t:= \mathbbm{1}_{v \in V(s_1)} )$ we use (with a different set of weights)
\begin{align*}
    i &= \sigma(W_{i,c}c^2_{1,2} + W_{i,t}t + b_{i}) \\
    f &= \sigma(W_{f,c}c^2_{1,2} + W_{f,t}t + b_{f}) \\
    g &= \tanh(W_{g,c}c^2_{1,2} + W_{g,t}t +b_{g}) \\
    o &= \sigma(W_{o,c}c^2_{1,2} + W_{o,t}t + b_{o}) \\
    h_v &= f * h_v + i * g 
    %c^*_1 &= f * c_{1,2}+ i * g \\
    %x_1 &= o * \tanh(f * h_{1,2} + i * g) \\
\end{align*} 
Where $t=\mathbbm{1}_{v \in V(s_1)}$.
Intuitively, we make it easy for the label to flow through. %one encoding for the component and for each node -- the node's significance or place in the component.

%%%%%%%%%%%%%%%%%%%%%%%%%%%%%%%
%%%%%%%%%%%%%%%%%%%%%%%%%%%%%%%
%%%%%%%%%%%%%%%%%%%%%%%%%%%%%%%
%%%%%%%%%%%%%%%%%%%%%%%%%%%%%%%
%%%%%%%%%%%%%%%%%%%%%%%%%%%%%%%
%%%%%%%%%%%%%%%%%%%%%%%%%%%%%%%

\section{Experiments}

% \subsection{Class-Redundancy and Levels}
% \label{appendix:classredundancylevels}

\subsection{Synthetic Graphs}

The ordering of the nodes of a graph $G$ are randomly shuffled before $G$ is feed to NPA and the output depends to some extent on this order. This makes it hard for a NN to overfit to the features that NPA produces on a training set. For datasets where the class-redundancy is large (e.g regular graphs) NPA might never produce the same encoding between the gradient steps and the training accuracy evaluation. This may cause NNs to overfit to the encodings NPA produces during the batch updates and underfit the encodings produced for evaluation of training accuracy. Even during training, NPA (and NPBA) might never produce the same representation for the same graph twice.

\subsection{Experiment Details}

We try and compare algorithms at the task of classifying graphs. Every dataset maps each of its graphs to a ground-truth class out of two possible classes.

We report the average and standard deviation of validation accuracies across the 10 folds within the cross-validation. We use the Adam optimizer with initial learning rate 0.01 and decay the learning rate by 0.5 every 50 epochs. We tune the number of epochs as a hyper-parameter, i.e., a single epoch with the best cross-validation accuracy averaged over the 10 folds was selected.
%Note that due to the small dataset sizes, an alternative setting, where hyper-parameter selection is done using a validation set, is extremely unstable, e.g., for MUTAG, the validation set only contains 18 data points.

In the experiments, the $W(G)$ features are summed and passed to a classify-NN consisting of either one fully-connected layer and a readout layer (for MUTAG, PTC, and PROTEINS) or two fully-connected layers and a readout layer (for NCI1), where the hidden-dim of the fully connected layers is of size $d_{hidden}$. For $h_{init}$ we use a linear-layer followed by a batchnorm (for MUTAG, PTC, and PROTINES) or a linear-layer followed by activation function and batchnorm (for NCI1). In addition, for NCI1 we used dropout=0.2 after each layer in the classify-net and on the vectors of $W(G)$ before summing them.

Also, in our experiments we skipped including the $w_i$ features for the single nodes. In fact, all datasets consist of connected graphs.

For the NPBA tree-lstm the dimensions of $c^1$ and $c^2$ is $d_{hidden}$. For the NPA the dimensions of $c^1$ and $c^2$ is $d_{hidden}$ and the dimension of $h$ is $d_{hidden}/2$. 

We used the following settings for $d_{hidden}$ and batch size:
\begin{itemize}
    \item PTC, PROTEINS, and MUTAG we used $d_{hidden}=16$, and batch-size=$32$.
    \item NCI1 we used $d_{hidden}=64$, and batch-size=$128$.
\end{itemize}

\end{document}